\documentclass[journal]{IEEEtran}
\usepackage{cite}

\usepackage[utf8]{inputenc}
\usepackage{amsthm}
\usepackage[bookmarks=true]{hyperref}
\usepackage{xcolor}
\usepackage{graphicx}
\usepackage{dsfont} %% Enables \mat
\usepackage{amsmath}
\usepackage{amssymb}
\usepackage{amsfonts,amsthm}       % blackboard math symbols
\usepackage{amsmath,amsbsy,amssymb,mathtools}
\usepackage{nicefrac}       % compact symbols for 1/2, etc.
\usepackage{microtype}      % microtypography
\usepackage{tikz}

\newtheorem{remark}{Remark}
\newtheorem{problem}{Problem}
\usepackage{url}
\usepackage{color}
\usepackage{subcaption}
\usepackage[utf8]{inputenc}
\usepackage[T1]{fontenc}%\usepackage{clomment}
\usepackage{balance}
\usepackage{algorithmic}
\usepackage[ruled,vlined, linesnumbered]{algorithm2e}

\usepackage[capitalize]{cleveref}
\crefformat{equation}{(#2#1#3)}

\newtheorem{definition}{Definition}
\newtheorem{lemma}{Lemma}

\newtheorem{theorem}{Theorem}
\newtheorem{proposition}{Proposition}
\usepackage{comment}
\usepackage{mathrsfs}
\usepackage[textfont=md,font=footnotesize]{caption}

\DeclareMathOperator*{\argmax}{arg\,max}

\newcommand{\doi}[1]{\href{http://dx.doi.org/#1}{\normalsize{\textsc{doi:}}~\nolinkurl{#1}}}
\newcommand{\arxiv}[1]{\href{http://arxiv.org/abs/#1}{\normalsize{\textsc{arxiv:}}~\nolinkurl{#1}}}

\renewcommand{\epsilon}{\varepsilon}

\let\originalleft\left
\let\originalright\right
\renewcommand{\left}{\mathopen{}\mathclose\bgroup\originalleft}
\renewcommand{\right}{\aftergroup\egroup\originalright}

\def\clap#1{\hbox to 0pt{\hss#1\hss}}

\definecolor{darkmidnightblue}{rgb}{0.0, 0.2, 0.4}
\DeclareMathAlphabet{\mathpzc}{OT1}{pzc}{m}{it}

\newcommand{\dfknote}[1]%
    {\textcolor{orange}{\textbf{DFK: #1}}}
\newcommand{\jlnote}[1]%
    {\textcolor{cyan}{\textbf{JL: #1}}}
    
\title{%Reach-Avoid Zero-Sum Games via Conservative Q-learning and its Applications in Learning Viability Kernel and Backward Reachable Sets
Inifinite Horizon Reach-Avoid Zero-Sum Games via Deep Reinforcement Learning
}
\author{Jingqi Li, Donggun Lee, Somayeh Sojoudi, Claire J. Tomlin
\thanks{The authors are with the University of California, Berkeley. Correspondence to \href{mailto:jingqili@berkeley.edu}{\tt jingqili@berkeley.edu}.\newline \indent \color{black}This research is supported by an NSF CAREER award, the Air Force Office of Scientific Research (AFOSR), NSF's CPS FORCES and VeHICaL projects, the UC-Philippine-California Advanced Research Institute, the ONR BRC grant for Multibody Systems Analysis, a DARPA Assured Autonomy grant, and the SRC CONIX Center.}
\thanks{This is a preprint manuscript.}
}

\begin{document}
% \begin{titlepage}
% ©2022 IEEE. Personal use of this material is permitted. Permission from IEEE must be obtained for all other uses, in any current or future media, including reprinting/republishing this material for advertising or promotional purposes, creating new collective works, for resale or redistribution to servers or lists, or reuse of any copyrighted component of this work in other works.

% This work has been submitted to the IEEE for possible publication. Copyright may be transferred without notice, after which this version may no longer be accessible.
% \end{titlepage}
\maketitle

\begin{abstract}
In this paper, we consider the \emph{infinite-horizon reach-avoid zero-sum game} problem, where the goal is to find a set in the state space, referred to as the \emph{reach-avoid set}, such that the system starting at a state therein could be controlled to reach a given target set without violating constraints under the worst-case disturbance. We address this problem by designing a new Lipschitz continuous value function with a contracting Bellman backup, where the super-zero level set %, i.e., the set of states where the value function is evaluated to be non-negative, 
recovers the reach-avoid set. Building upon this, we prove that the proposed method can be adapted to compute the \emph{viability kernel}, or the set of states which could be controlled to satisfy given constraints, and the \emph{backward reachable set}, or the set of states that could be driven towards a given target set. Finally, we propose to alleviate the curse of dimensionality issue in high-dimensional problems by extending \emph{Conservative Q-Learning}, a deep reinforcement learning technique, to learn a value function such that the super-zero level set of the learned value function serves as a (conservative) approximation to the reach-avoid set. %Our theoretical and empirical results suggest that the proposed method could learn reliably the reach-avoid set and the optimal control policy even with neural network approximation.

\end{abstract}
\section{Introduction}
\label{sec:introduction}

Ensuring the safety and performance of autonomous systems is essential for safety-critical systems, such as autonomous driving \cite{muhammad2020deep}, surgical robots \cite{thananjeyan2020safety}, and air traffic control \cite{hwang2007protocol}. %, human-robot interaction: \cite{lyu2019human}, human-robot collaboration \cite{villani2018survey}, environmental monitoring robots \cite{dunbabin2012robots}, and battery safety management \cite{wang2018review}.
Some safe control tasks could be modeled as driving the system's state to a target set in the state space while satisfying certain safety constraints \cite{bansal2017hamilton}. This is referred to as the \emph{reach-avoid} problem \cite{summers2011stochastic}. However, it is challenging to solve this problem under the presence of uncertainty \cite{dorato1987historical}, such as modeling errors and environmental disturbances. One way to accommodate this is to consider the reach-avoid problem under the worst-case scenario, which could be formulated as a zero-sum game between the control inputs and an adversary, accounting for the uncertainty or disturbance, with an objective to compromise the efforts of control inputs\cite{tomlin1998conflict}. 

In the general case, this problem is challenging because it requires solving a zero-sum game with nonlinear dynamics, where the objective involves the worst-case performance instead of an average or cumulative performance over time, as in most constrained reinforcement learning (RL) papers \cite{achiam2017constrained,chow2017risk,altman1999constrained}. Moreover, no future constraint violation is considered once the state trajectory enters the target set in the reach-avoid zero-sum game. Since constrained RL aims to learn a policy satisfying constraints at all times, it can only learn a conservative sub-optimal policy for the reach-avoid zero-sum game.

% In this paper, we consider an infinite-horizon reach-avoid problem with instantaneous safety constraints. We adapt the classical Augmented Lagrangian method \cite{nocedal2006numerical} to obtain a safe policy satisfying instantaneous safety constraints. Our work is closely related to \cite{liu2019ipo}, where an interior-point method is adapted to solve the safe RL problem. One major difference is that we relax the assumption of an initial safe policy, which is required in \cite{liu2019ipo}. 

A well-known approach to solving the reach-avoid zero-sum game is the Hamilton-Jacobi (HJ) method \cite{margellos2011hamilton,fisac2015reach}, in which one can design a value function such that the sign of the value function evaluated at a state encodes the safety and performance information of that state. In addition, reach-avoid games with multiple agents \cite{selvakumar2019feedback} and stochastic systems \cite{vidal2002probabilistic} have also been studied. 
These methods provide a closed-loop control policy for the continuous-time and finite-horizon zero-sum game setting, where an agent is required to reach a target set safely within a given finite-time horizon. 

However, there are two major limitations of existing work on HJ-based methods for finite-horizon reach-avoid games. Firstly, for complex or uncertain dynamical systems, it is difficult to predict a sufficient time horizon to guarantee the feasibility of the finite-time reach-avoid zero-sum game problem. This motivates the formulation of the infinite-horizon reach-avoid zero-sum game, where we remove the requirement of reaching the target set within a pre-specified finite horizon. The infinite-horizon reach-avoid game has been studied in only a few prior works. Among these is \cite{hsusafety}, in which the reach-avoid problem without an adversary player is considered. The authors introduce the time-discount factor, a parameter discounting the impact of future reward and constraints, to the design of a contractive Bellman backup such that by annealing the time-discount factor to 1 they can obtain the reach-avoid set. Empirical results in \cite{hsusafety} suggest the potential of this method in solving the infinite-horizon reach-avoid zero-sum~games.

A second limitation is the computational complexity of classical HJ-based methods, which are exponential in the dimension of the state space \cite{bansal2017safe}. Several approaches have been proposed in the literature to alleviate this. A line of work \cite{mitchell2008flexible,landry2018reach,singh2018robust} formulate the reach-avoid zero-sum games as polynomial optimization problems for which there is no need to grid the entire state space. Another line of work on system decomposition, that is, decomposing a high-dimensional control problem into low-dimensional problems, is also a promising direction \cite{chen2016exact,chen2018decomposition}.
\textcolor{black}{Underapproximation of the reach-avoid set is proposed in \cite{gleason2017underapproximation} and an efficient open-loop policy method is developed in \cite{zhou2012general}. }
Nevertheless, these approaches presume certain structural assumptions about the dynamical systems and costs, which restrict their applications.%but the price we pay is that this method only work for some particular systems and for general systems it is only an approximation. 

With the power of neural networks, deep RL has been proven to be a promising technique for high-dimensional optimal control tasks \cite{mnih2013playing,silver2014deterministic,haarnoja2018soft}, where a policy is derived to maximize the accumulated reward at each time step. It has been shown in \cite{munos1999gradient} that neural networks could be leveraged to approximate the value function in high-dimensional optimal control problems and therefore alleviate the curse of dimensionality of Hamilton-Jacobi reachability analysis. For example, sinusoidal neural networks is proven to be a good functional approximator for learning the value function of backward-reachable-tube problem for high-dimensional dynamical systems \cite{bansal2020deepreach}. A deep reinforcement learning-based method is proposed in \cite{fisac2019bridging} to learn a neural network value function for \emph{viability kernel}, the set of initial states from which a state trajectory could be maintained to satisfy pre-specified constraints. In particular, the paper \cite{fisac2019bridging} designs a contractive Bellman backup and learns a conservative approximation to the viability kernel. This method is further extended to solve infinite-horizon reach-avoid problem in \cite{hsusafety}.

In this paper, we propose a new HJ-based method for the infinite-horizon reach-avoid zero-sum game, and we develop a deep reinforcement learning method to alleviate the curse of dimensionality in solving it. Our contributions are threefold. We first propose a new value function for the reach-avoid zero-sum game, where the induced Bellman backup is a contraction mapping and ensures the convergence of value iteration to the unique fixed point. Subsequently, we analyze the designed value function by first proving its Lipschitz continuity, and then show that the new value function could be adapted to compute the viability kernel and the \emph{backward reachable set}, that is, the set of states from which a state trajectory could be controlled to reach the given target set. Finally, we alleviate the curse of dimensionality issue by proposing a deep RL algorithm, where we extend conservative Q-learning \cite{kumar2020conservative} to learn the value function of the infinite-horizon reach-avoid zero-sum game. We obtain upper and lower bounds for the learned value function. Our empirical results suggest that a (conservative) approximation of the reach-avoid set can be learned, even with neural network approximation.

Unlike the previous work \cite{hsusafety} where the value function is discontinuous in general and the reach-avoid set can only be obtained when the time-discount factor is annealing to 1, the proposed value function is Lipschitz continuous under certain conditions and can exactly recover the reach-avoid set without annealing the time-discount factor. The obtained reach-avoid set is agnostic to any time-discount factors in the interval $[0,1)$.
%the super-zero level set of it, i.e., the set of states where the value function are evaluated to be non-negative, exactly recovers the reach-avoid set. The obtained reach-avoid set is agnostic to any time-discount factors in the interval $[0,1)$.} %Moreover, the proposed value function is Lipschitz continuous under certain conditions, while the value function in \cite{hsusafety} is discontinuous in general.
%{\color{blue}Unlike the previous work \cite{hsusafety} where the time-discount factor, a parameter discounting the importance of future reward and constraint, is used to design a contractive Bellman backup, we directly incorporate the time-discount factor into the design of value function. This new approach has two advantages over the one in \cite{hsusafety}. Firstly, the induced Bellman backup is a contraction mapping. Moreover, the super-zero level set of the proposed value function, i.e., the set of states where the value function are evaluated to be non-negative, exactly recovers the reach-avoid set. The obtained reach-avoid set is agnostic to any time-discount factors in the interval $[0,1)$.} 
In addition, this new value function could be adapted to compute the viability kernel and backward reachable set, which constitutes a unified theoretical perspective on reachability analysis concepts such as reach-avoid set, viability kernel, and backward reachable set.

The rest of the paper is organized as follows. %In Section \ref{subsec:prelim}, we recall several crucial preliminaries. 
In Section \ref{sec:problem_formulation}, we formulate the infinite-horizon reach avoid problem. %In Section \ref{sec:related_work}, we briefly review closely related methods in the literature. 
We present our main theoretical results in Sections \ref{sec:methods} - \ref{sec:pure_constraint_and_pure_target}, with proofs provided in the Appendix. In Section \ref{sec:experiments}, we illustrate our algorithm in multiple experiments. Finally, we conclude the paper in Section \ref{sec:conclusion}.

\section{Problem Formulation}
\label{sec:problem_formulation}
We denote the state of a system by $x\in \mathbb{R}^n$. Define $\mathbf{u}:=\{u_0,u_1,\dots\}$ and $\mathbf{d} := \{d_0,d_1,\dots\}$, where $u_t\in\mathcal{U}\subseteq \mathbb{R}^m$ and $d_t\in \mathcal{D}\subseteq \mathbb{R}^\ell$ are a sequence of control inputs and a sequence of disturbances, respectively. We assume $\mathcal{U}$ and $\mathcal{D}$ to be compact sets. Let $\xi_{x}^{\mathbf{u},\mathbf{d}}(t)$ be the state trajectory evaluated at time $t$, which evolves according to the update rule
\begin{equation}
\begin{aligned}
    \xi_x^{\mathbf{u},\mathbf{d}}(0) & = x,\\
    \xi_x^{\mathbf{u},\mathbf{d}}(t+1) &= f(\xi_x^{\mathbf{u},\mathbf{d}}(t), u_t,d_t),\ t\in\mathbb{Z}_+,
\end{aligned}
\end{equation}
where $\mathbb{Z}_+$ is the set of all non-negative integers and the system dynamics $f(\cdot,\cdot,\cdot ):\mathbb{R}^n\times \mathcal{U}\times \mathcal{D}\to \mathbb{R}^n$ is assumed to be Lipschitz continuous in the state.

Many control problems involving safety-critical systems can be interpreted as driving the system's state to a specific region while satisfying certain safety constraints. This intuition can be formalized by introducing the concept of \emph{target set} as the set of desirable states to reach as well as the concept of \emph{constraint set} as the set of all the states complying with the given constraints. The task is to design control inputs such that even under the worst-case disturbance the system state trajectory hits the target set while staying in the constraint set all times. We represent the target set and constraint set by $\mathcal{T}=\{x\in\mathbb{R}^n: r(x)>0\}$ and $\mathcal{C}=\{x\in\mathbb{R}^n: c(x)>0 \}$, respectively, for some Lipschitz continuous and bounded functions $r(\cdot): \mathbb{R}^n \to \mathbb{R}$ and $c(\cdot):\mathbb{R}^n\to \mathbb{R}$.

At each time instance $t\in\mathbb{Z}_+$, the control input $u_t$ aims to move the state towards the target set while staying inside the constraint set, whereas the disturbance $d_t$ attempts to drive the state away from the target set and the constraint set. In this work, we consider a conservative setting, where at each time instance the control input plays before the disturbance. %Under this setting, the disturbance takes its action by considering the action of the control inputs, which allows more advantages to the disturbance than the control inputs and therefore constitutes a sequence of worst-case disturbance. 
We propose to account for this playing order by adopting the notion of \emph{non-anticipative strategy}:%, which is defined below:%, which constraints the disturbance to use only history control inputs but no future information, and define the set $\Phi$ of non-anticipative strategies for the disturbance \cite{Evans1984} as,
\begin{definition}[Non-anticipative strategy]
A map $\phi: \{u_t\}_{t=0}^\infty\to \{d_t\}_{t=0}^\infty$, where $u_t\in\mathcal{U}$ and $d_t\in\mathcal{D}$, $\forall t\in\mathbb{Z}_+$, is a non-anticipative strategy if it satisfies the following condition: Let $\{u_t\}_{t=0}^\infty$ and $\{\bar{u}_t\}_{t=0}^\infty$ be two sequences of control inputs. Let $\{d_t\}_{t=0}^\infty=\phi(\{u_t\}_{t=0}^\infty)$ and $\{\bar{d}_t\}_{t=0}^\infty = \phi(\{\bar{u}_t\}_{t=0}^\infty)$. For all $T\ge 0$, if $u_t=\bar{u}_t$, $\forall t\in\{0,1,\dots,T\}$, then $d_t = \bar{d}_t$, $\forall t\in \{0,1,\dots, T\}$.
\end{definition}
%\begin{equation}\color{blue}
%    \begin{aligned}
%    \Phi := \{ \phi: &\{u_t\}_{t=0}^\infty \to \{d_t\}_{t=0}^\infty ~ | ~ \forall T, \textrm{ if for }\{u_t\}_{t=0}^\infty\textrm{ and }\\&\{\bar{u}_t\}_{t=0}^\infty, u_t = \bar{u}_t,\ t\in[0,T]\ \textrm{, then for }\\ &\{d_t\}_{t=0}^\infty = \phi(\{u_t\}_{t=0}^\infty) \textrm{ and }\{\bar{d}_t\}_{t=0}^\infty = \phi(\{\bar{u}_t\}_{t=0}^\infty),\\ &d_t = \bar{d}_t ,t\in[0,T]\},
%    \end{aligned}
%\end{equation}
The intuition behind this non-anticipative strategy is that, given two sequences of control inputs sharing the same values for the first $T$ entries, the corresponding entries of the two sequences of disturbance must also be the same, meaning the disturbance cannot use any information about the future control inputs when taking its own action. We denote by $\Phi$ the set of all non-anticipative strategies. Building upon the previous discussion on reach-avoid zero-sum games, we are now ready to define the definition the \emph{reach-avoid set}:
\begin{equation*}
    \begin{aligned}
    \mathcal{RA(T,C)} := &\{x\in\mathbb{R}^n: \forall \phi\in\Phi, \exists \{u_t\}_{t=0}^\infty \textrm{ and } T\ge 0,\\ &\textrm{s.t., } \forall t\in[0,T],\xi_{x}^{\mathbf{u},\phi(\mathbf{u})}(T) \in\mathcal{T}\wedge   \xi_{x}^{\mathbf{u},\phi(\mathbf{u})}(t)\in \mathcal{C}\},
    \end{aligned}
\end{equation*}
where 's.t.' is the abbreviation of the phrase 'such that'. As such, we formalize the problem to be studied in this paper below.
\begin{problem}[Infinite-horizon reach-avoid game]\label{Reach_avoid_problem}
 Let $\mathcal{T}:=\{x\in\mathbb{}{R}^n: r(x)>0\}$ and $\mathcal{C}:=\{x\in\mathbb{R}^n:c(x)>0\}$ be the target and the constraint set, respectively. Find the reach-aviod set $\mathcal{RA(T,C)}$ and an optimal sequence of control inputs $\{u_t\}_{t=0}^\infty$ for each state in $\mathcal{RA(T,C)}$ such that under the worst-case disturbance $\{d_t\}_{t=0}^\infty$ the state trajectory will reach the target set without violating the constraints. 
\end{problem}

\section{A New Value Function for Infinite-horizon Reach-Avoid Games}
In this section, we address Problem~\ref{Reach_avoid_problem} by introducing a new Lipschitz continuous value function such that its sign evaluated at a state indicates whether that state can be driven to the target set without violating constraints. To be more specific, we will show that the \emph{super-zero level set} of the designed value function recovers the reach-avoid set in Problem~\ref{Reach_avoid_problem}. Subsequently, we derive a Bellman backup equation. The fixed-point iteration based on this Bellman backup equation induces a unique fixed point solution, which will be proven to be the designed value function \eqref{eq:inf_horizon_reach_avoid_problem}. Finally, we will show that the designed value function is Lipschitz continuous, which is a favorable property for HJ analysis \cite{basco2019lipschitz}.

Suppose that a state trajectory reaches a target set $\mathcal{T}$ at time $t$, i.e., $\xi_{x}^{\mathbf{u,d}}(t)\in\mathcal{T}$, while being maintained within the constraint set $\mathcal{C}$ for all $\tau \in\{0,1,\dots ,t\}$ under an arbitrary sequence of disturbances $\mathbf{d}=\{d_t\}_{t=0}^\infty$. Let $\gamma\in[0,1)$ be the time-discount factor, which discounts the impact of future reward and constraints. We observe that such a sequence of control inputs $\mathbf{u} = \{u_t\}_{t=0}^\infty$ should satisfy
\begin{equation*}
\begin{aligned}
    \max_{u_0}\min_{d_0}&\max_{u_1}\min_{d_1}\cdots \max_{u_t}\min_{d_t} \min \big\{ \gamma^t r(\xi_{x}^{\mathbf{u,d}}(t)), \\ &\min_{\tau=0,\dots, t} \gamma^{\tau} c(\xi_{x}^{\mathbf{u,d}}(\tau)) \big\}>0.
\end{aligned}
\end{equation*}
Therefore, to verify the existence of a sequence of control inputs driving the state trajectory to the target set safely from an initial state $x_0$, it suffices to check the sign of the following term:
\begin{equation}\label{eq:inf_time_check_finite_time}
    \begin{aligned}
        \max_{u_0}\min_{d_0}\max_{u_1}&\min_{d_1}\dots \sup_{t=0,\dots} \min\big\{ \gamma^t r(\xi_{x}^{\mathbf{u,d}}(t)),\\& \min_{\tau = 0,\dots,t}\gamma^\tau c(\xi_{x}^{\mathbf{u,d}}(\tau)) \big\}.
    \end{aligned}
\end{equation}
\begin{comment}
\textcolor{blue}{Instead of using this strategy, I suggest
\begin{align*}
    \max_{u_0}\min_{d_0}\max_{u_1}\min_{d_1}\dots \max_{t=0,\dots, \infty} \min\big\{ \gamma^t r(\xi_{x}^{\mathbf{u,d}}(t)), \min_{\tau = 0,\dots,t}\gamma^\tau c(\xi_{x}^{\mathbf{u,d}}(\tau)) \big\} \\
    = \left(\prod_{t=0}^\infty \max_{u_t}\min_{d_t}\right) \max_{t=0,\dots, \infty} \min\big\{ \gamma^t r(\xi_{x}^{\mathbf{u,d}}(t)), \min_{\tau = 0,\dots,t}\gamma^\tau c(\xi_{x}^{\mathbf{u,d}}(\tau)) \big\}.
\end{align*}
}
\end{comment}
By using the notion of the non-anticipative strategy $\phi$, we can simplify \eqref{eq:inf_time_check_finite_time} to the term
\begin{equation*}
\begin{aligned}
    \inf_{\phi}\max_{\mathbf{u}}&\sup_{t=0,\dots, } \min \big\{ \gamma^t r(\xi_{x}^{\mathbf{u},\phi(\mathbf{u})}(t)), \min_{\tau=0,\dots, t} \gamma^{\tau} c(\xi_{x}^{\mathbf{u},\phi(\mathbf{u})}(\tau)) \big\}.
\end{aligned}
\end{equation*}
This implies that a new value function for Problem~\ref{Reach_avoid_problem} can be defined~as
\begin{equation}\label{eq:inf_horizon_reach_avoid_problem}
\begin{aligned}
    V(x) := \inf_{\phi} &\max_{\mathbf{u}} \sup_{t=0,\dots} \min \big\{ \gamma^t r(\xi_{x}^{\mathbf{u},\phi(\mathbf{u})}(t)) ,\\& \min_{\tau=0,\dots,t} \gamma^{\tau} c(\xi_{x}^{\mathbf{u},\phi(\mathbf{u})}(\tau)) \big\},\ \forall x\in\mathbb{R}^n.
\end{aligned}
\end{equation}
The sign of the value function  encodes some crucial safety information about each state, that is, whether the state could be driven towards the target set while satisfying the given constraints. More precisely, $V(x)>0$ if for every disturbance strategy $\phi$, there exist $\mathbf{u}=\{u_t\}_{t=0}^\infty$ and $T<\infty$ such that $\xi_{x}^{\mathbf{u},\phi(\mathbf{u})}(T)\in\mathcal{T}$ and $\xi_{x}^{\mathbf{u},\phi(\mathbf{u})} (t)\in\mathcal{C}$ for all $ t\in[0,T]$. Based on this intuition, we characterize the relationship between the super-zero level set of the value function~\eqref{eq:inf_horizon_reach_avoid_problem} and the reach-avoid set next.

\begin{theorem}\label{thm:V>0}
Consider the value function $V(x)$ defined in \eqref{eq:inf_horizon_reach_avoid_problem}. For every $\gamma\in[0,1)$, it holds that
$\{x\in\mathbb{R}^n:V(x)>0\} = \mathcal{RA(T,C)}$.
\end{theorem}
% \begin{remark}
% The definition in \cite{hsusafety} is problematic. Consider the example $x_{t+1} = 0.9x_t + 0*u_t$, where .
% \end{remark}

We propose to compute the value function \eqref{eq:inf_horizon_reach_avoid_problem} using dynamic programming. To this end, a contractive Bellman backup will be derived below.
\label{sec:methods}

\begin{theorem}\label{thm:Bellman}
Let $\gamma\in[0,1)$ be the time-discount factor. Suppose $U:\mathbb{R}^n\to\mathbb{R}$ is a bounded function. Consider the Bellman backup $B[\cdot] $ defined as,
\begin{equation}\label{eq:bellman_backup}
\begin{aligned}
    B[U](x) \coloneqq &\min\big\{ c(x), \max\{ r(x), \gamma \max_u \min_{d} U(f(x,u,d)) \}\big\}.
\end{aligned}
\end{equation}
Then, \eqref{eq:inf_horizon_reach_avoid_problem} is the unique solution to the Bellman backup equation, i.e., $V=B[V]$, 
%\begin{equation}
%    V=B[V],\label{eq:bellman_backup_equation}
%\end{equation}
and $    \|B[V_1] - B[V_2]\|_\infty \le \gamma \| V_1-V_2 \|_\infty$, for all bounded functions $V_1$ and $V_2$.
\end{theorem}

With the Bellman backup \eqref{eq:bellman_backup}, we define the \emph{value iteration} as the following recursion, starting from an arbitrary bounded function $V^{(0)}(\cdot):\mathbb{R}^n \to \mathbb{R}$, 
\begin{equation}\label{eq:value_iteration}
V^{(k+1)} := B[V^{(k)}],\  \forall k\in\mathbb{Z}_+.
\end{equation}
Since the fixed-point iteration based method using a contraction mapping $B$ will guarantee the convergence to a fixed point within that bounded function space, the Bellman backup \eqref{eq:bellman_backup} ensures the convergence of value iteration to \eqref{eq:inf_horizon_reach_avoid_problem}, i.e., $\lim_{k\to \infty}V^{(k)} = V$.

Once we obtain the value function \eqref{eq:inf_horizon_reach_avoid_problem}, we can calculate the following Q function:
\begin{equation}\label{eq:Q}
    Q(x,u,d) \coloneqq \min\big\{ c(x), \max \{r(x), \gamma V(f(x,u,d))\} \big\},
\end{equation}
where $V(x) = \max_u\min_d Q(x,u,d)$. Given a state $x\in\mathbb{R}^n$, we can find an optimal control input and disturbance by $u^* = \arg\max_u \min _d Q(x,u,d)$ and $d^*=\arg\min_d Q(x,u,d)$, respectively. This provides a way to extract the optimal control and worst-case disturbance for Problem \ref{Reach_avoid_problem}. Therefore, the proposed value function \eqref{eq:inf_horizon_reach_avoid_problem} constitutes a solution to Problem~\ref{Reach_avoid_problem} in the sense that one can find the reach-avoid set, the optimal sequence of control inputs, and the worst-case sequence of disturbances.

In what follows, we characterize the Lipschitz continuity of the proposed value function \eqref{eq:inf_horizon_reach_avoid_problem}. We show that the Lipschitz continuity of the value function \eqref{eq:inf_horizon_reach_avoid_problem} could be ensured under certain conditions on the dynamics and the time-discount factor.

\begin{theorem}[Lipschitz continuity]\color{black}
Suppose that the bounded functions $r(\cdot)$ and $c(\cdot)$ are $L_r$- and $L_c$-Lipschitz continuous, respectively. Assume also that the dynamics $f(x,u,d)$ is $L_f$-Lipschitz continuous in $x$, for all $u\in\mathcal{U}$ and $d\in\mathcal{D}$. Let $L:= \max(L_r,L_c)$. Then, $V(x)$ is $L$-Lipschitz~continuous~if~$L_f \gamma <1$.
\label{lemma:continuity}
\end{theorem}
\begin{remark}
It is known that the sample complexity of neural network approximation could be improved if the function to be approximated is continuous \cite{gouk2021regularisation}. Therefore, the Lipschitz continuity of the value funtion \eqref{eq:inf_horizon_reach_avoid_problem} is a beneficial property yielding reliable empirical performance when we deploy neural networks approximation to handle curse of dimensionality as highlighted in Section~\ref{sec:6d experiment}.
\end{remark}

%{\color{green}We remark that the major difference between the Bellman equation \eqref{eq:bellman_backup} and the Discounted Reach-Avoid Bellman Equation (DRABE) in \cite{hsusafety} is that the super zero-level set of \eqref{eq:inf_horizon_reach_avoid_problem} is the exact reach-avoid set for problem \ref{Reach_avoid_problem}, which is agnostic to the time discount factor $\gamma\in[0,1)$. Also, for certain conditions as specified in Theorem \ref{lemma:continuity}, the value function \eqref{eq:inf_horizon_reach_avoid_problem} is continuous, while the value function in \cite{hsusafety} could be discontinuous.}

\section{Computing the Viability Kernel and the Backward Reachable Set}\label{sec:pure_constraint_and_pure_target}
In this section, we show that the designed value function \eqref{eq:inf_horizon_reach_avoid_problem} can be adapted to compute the viability kernel \cite{fisac2019bridging} and the backward reachable set \cite{bansal2020deepreach}. We first present the formal definitions of these notions below.

We adopt the same definition of viability kernel as in \cite{fisac2019bridging}, where the \emph{viability kernel} is defined with respect to a closed set $\bar{\mathcal{C}}=\{x\in\mathbb{R}^n: c (x)\ge 0\}$, i.e., the closure of the constraint set $\mathcal{C}$, as
\begin{equation*}
    \Omega(\mathcal{\bar{C}}):=\{x\in\mathbb{R}^n:\forall \phi\in\Phi , \exists \mathbf{u} , \textrm{s.t., }\xi_{x}^{\mathbf{u},\phi (\mathbf{u})}(t)\in \mathcal{\bar{C}},\forall t\ge 0  \},
\end{equation*}
which contains all states $x$ from which a state trajectory $\xi_x^{\mathbf{u},\phi(\mathbf{u})}$ can be drawn such that $\xi_{x}^{\mathbf{u},\phi(\mathbf
{u})}(t)\in\bar{\mathcal{C}}$, $\forall t>0$. This set is of importance when the control task is to make the system satisfy some safety-critical constraints in real-world applications \cite{fisac2019bridging}.

The \emph{backward-reachable set} \cite{bansal2020deepreach} of a target set $\mathcal{T}=\{x\in\mathbb{R}^n: r(x)>0\}$ is defined as
\begin{equation*}
\begin{aligned}
    \mathcal{R(T)} := \{ x\in\mathbb{R}^n: &\forall \phi\in\Phi,\exists \mathbf{u}, \exists t \ge 0,\textrm{s.t., } \xi_{x}^{\mathbf{u},\phi(\mathbf{u})}(t)\in\mathcal{T}\},
\end{aligned}
\end{equation*}
which includes all states that can be driven towards the target set $\mathcal{T}$ in finite time. The backward reachable set is a useful concept in reachability analysis \cite{bansal2020deepreach} because it specifies the set of states that can be controlled to reach a given target set. For instance, if an airplane is in the backward reachable set of another one, then it indicates that the two planes could collide with each other if they are not properly controlled.

In what follows, we show how the value function \eqref{eq:inf_horizon_reach_avoid_problem} can be adapted to compute the viability kernel and backward-reachable set.

\begin{proposition}\label{theorem:pure_safe}
Assume that $r(x)=-1$ for all $ x\in\mathbb{R}^n$. It holds that $V(x)\le 0$ for all $ x\in\mathbb{R}^n$. In addition, $V(x)=0$ if and only if $x\in\Omega(\mathcal{\bar{C}})$.
\end{proposition}%

\begin{remark}
The value function in Proposition \ref{theorem:pure_safe} could serve as a control barrier function \cite{ames2016control}, from which we can derive a policy keeping the system states within a given set of the state space.
\end{remark}

%\begin{remark}
%Proposition~\ref{theorem:pure_safe} shows that the zero-level set of the induced value function is the $\Omega(\bar{\mathcal{C}})$, where $\bar{\mathcal{C}}$ is the closure of the open set $\mathcal{C}$. For an open constraint set $\mathcal{C}=\{x\in \mathbb{R}^n:c(x)>0\}$, a conservative viability kernel with respect to $\mathcal{C}$ can be learned by substituting an $\epsilon$-conservative constraint function defined as $\tilde{c}(x):=c(x)-\epsilon$ into the constraint function in \eqref{eq:inf_horizon_reach_avoid_problem}, for some $\epsilon>0$, where $\epsilon$ controls the conservativeness of the approximation.
%\end{remark}

Similarly, we can compute the backward-reachable set by substituting particular constraint functions into \eqref{eq:inf_horizon_reach_avoid_problem}.

\begin{proposition}\label{theorem:pure_reach}
Assume that $c(x)=1$ for all $x\in\mathbb{R}^n$. It holds that $V(x)\ge 0$ for all $x\in\mathbb{R}^n$. In addition, $V(x) >  0$ if and only if $x\in \mathcal{R(T)}$.
\end{proposition}

Given the above results, our method provides a new theoretical angle to computing the reach-avoid set, viability kernel, and backward reachable set for infinite-horizon games, where a common value function \eqref{eq:inf_horizon_reach_avoid_problem} can be adaptively used to compute multiple important sets for safety-critical analysis.

\section{A Deep Reinforcement Learning Algorithm}
%In the preceding sections, we introduced the new value function \eqref{eq:inf_horizon_reach_avoid_problem} for Problem \ref{Reach_avoid_problem} and derived the Bellman backup \eqref{eq:bellman_backup}. 
%Classical HJ methods involve first griding the entire state space and then conducting value iteration over the grid. However, the computational complexity is exponential in the dimension of the state space, leading to the curse of dimensionality issue \cite{bansal2017safe}. 
In this section, we develop a deep RL algorithm, which alleviates the curse of dimensionality issue for high dimensional problems. In particular, since the satisfaction of given safety constraints is vital for safety-critical systems, we propose to extend conservative Q-learning (CQL) \cite{kumar2020conservative} to Problem \ref{Reach_avoid_problem} and develop a deep RL algorithm for solving it, where the conservatism is favorable due to the neural network approximation error. More precisely, we learn a value function parameterized by a neural network which in theory is a lower bound of the value function \eqref{eq:inf_horizon_reach_avoid_problem}. The super-zero level set of the converged value function in CQL is a subset of $\mathcal{RA(T,C)}$, i.e., a conservative approximation of $\mathcal{RA(T,C)}$.%, where the conservatism is favorable due to the neural networks approximation error.

The reason why CQL manages to learn a value function lower bound \eqref{eq:inf_horizon_reach_avoid_problem} is that CQL minimizes not only the Bellman backup error but also the value function itself at each iteration \cite{kumar2020conservative}. %{\color{blue}Minimizing the Bellman backup error can lead us to get closer to the value function \eqref{eq:inf_horizon_reach_avoid_problem}}, while the extra effort paid for minimizing the Q function yields a Q function that lower bounds the nominal Q function in some sense. 
Since the super-zero level set of the value function \eqref{eq:inf_horizon_reach_avoid_problem} is the reach-avoid set $\mathcal{RA(T,C)}$, the super-zero level set of a function that is a lower bound of \eqref{eq:inf_horizon_reach_avoid_problem} is a subset of $\mathcal{RA(T,C)}$.%, i.e., a conservative approximation of~$\mathcal{RA(T,C)}$.

Similar to prior works on Deep RL \cite{mnih2013playing,silver2014deterministic,kumar2020conservative}, we approximate the Q function \eqref{eq:Q} by a neural network $Q_\theta(x,u,d):\mathbb{R}^n\times \mathcal{U}\times \mathcal{D}\to \mathbb{R} $, where $\theta\in\mathbb{R}^M$ is the vector of parameters of the neural network and $M$ is the total number of parameters. {We define the neural network value function as $V_\theta(x) := \max_u\min_d Q_\theta(x,u,d)$, $\forall x\in\mathbb{R}^n$.} We adopt the CQL framework \cite{kumar2020conservative} to learn the value function \eqref{eq:inf_horizon_reach_avoid_problem} by replacing the Bellman backup therein with \eqref{eq:bellman_backup}. With $\lambda\ge 0$, we extend the loss function in \cite{kumar2020conservative} to the reach-avoid zero-sum game setting and propose the following loss function $L(\theta):\theta\in\mathbb{R}^M \to\mathbb{R} $ for the neural network parameters:
\begin{equation}\label{eq:cql_loss}
	\begin{aligned}
		L(\theta) :=&\mathbb{E}_{x\sim \mu}\big[ \|  V_\theta(x)  - B[ V_\theta](x) \|_2^2 + \lambda V_\theta(x) \big],
	\end{aligned}
\end{equation}
where $\mu$ is the uniform distribution over a compact set $\mathcal{X}\subseteq \mathbb{R}^n$ and the set $\mathcal{X}$ is subject to user's choice. %{\color{red}Ideally we should uniformly sample every state $x\in \mathbb{R}^n$ when evaluating the loss function \eqref{eq:cql_loss}. Since there does not exist a uniform distribution over the entire set $R^n$, we consider learning compact subset $\mathcal{X}\subseteq \mathbb{R}^n$}

As suggested in \cite{kumar2020conservative}, given a Bellman backup \eqref{eq:bellman_backup}, its corresponding CQL Bellman backup $B_{CQL}[\cdot]$ can be derived as
\begin{equation}\label{eq:CQL_Bellman_backup}
	B_{CQL}[U](x) := B[U](x) - \lambda,
\end{equation}
where $U:\mathbb{R}^n\to \mathbb{R}$ is a bounded function and $B[\cdot]$ is the Bellman backup \eqref{eq:bellman_backup}. One can show that this Bellman backup is a contraction mapping by a similar reasoning as in the proof of Theorem~\ref{thm:Bellman}.

Given a bounded functional space, the fixed-point theorem \cite{agarwal2001fixed} ensures that the value iteration based on \eqref{eq:CQL_Bellman_backup} will converge to a unique value function, which we refer to as $V_{CQL}(x)$. We characterize the relationship between the nominal value function $V(x)$ in \eqref{eq:inf_horizon_reach_avoid_problem} and its conservative counterpart $V_{CQL}(x)$ below.
\begin{theorem}\label{thm:sandwidch}
	Let $V(x)$ and $V_{CQL}(x)$ be the nominal value function defined in \eqref{eq:inf_horizon_reach_avoid_problem} and CQL value function, respectively. We have
	\begin{equation}
		V(x) - \frac{\lambda}{1-\gamma}\le V_{CQL}(x) \le V(x)-\lambda,\ \forall x \in \mathbb{R}^n.
	\end{equation}
\end{theorem}

Theorem~\ref{thm:sandwidch} implies that the value function learned by CQL would be a lower bound of the true value function \eqref{eq:inf_horizon_reach_avoid_problem}. However, the learned value function from CQL would still be bounded, and it is not too pessimistic, i.e., no state will have a value going to negative infinity. We remark here that the bounds in Theorem~\ref{thm:sandwidch} are specialized for the reach-avoid zero-sum game.

Building on the above results, we propose Algorithm \ref{alg:DQN}, where we substitute the Bellman backup in the algorithm of CQL \cite{kumar2020conservative} with \eqref{eq:bellman_backup}. As in DQN \cite{mnih2013playing}, we adopt \emph{experience replay} in Algorithm~\ref{alg:DQN}, where we store the state transition $\psi_t=(\xi_x^{\mathbf{u},\mathbf{d}}(t),u_t,d_t,\xi_x^{\mathbf{u},\mathbf{d}}(t+1))$ at each time step in a replay memory data-set $\Psi=\{\psi_0,\psi_1,\dots\}$. Subsequently, we apply Q-learning updates to random samples from $\Psi$ in step \ref{alg:update_gradient} of Algorithm \ref{alg:DQN}. It is shown that applying Q-learning to random samples drawn from experience replay could break down the correlation between consecutive transitions and recall rare transitions, which leads to a better convergence performance \cite{mnih2013playing}. %{\color{red}In addition, it has been shown that DQN with quantile-regression \cite{dabney2018distributional} usually leads to robust training and test performance. Motivated by this, in the following experiment section, we implement Algorithm \ref{alg:DQN} also under quantile regression.}

\begin{algorithm}[t]
	\SetAlgoLined
	\caption{Conservative Reach-Avoid Deep Q-learning}\label{alg:DQN}
	%Initialize replay memory $\Psi$ to capacity $N$ and 
	Initialize the action-value function $Q_{\theta_0}(x,u,d)$ with random weights $\theta_0$. Select a stepsize $\alpha>0$, the total number of iterations $K$, the batch size $J$, and the rollout horizon $T$\;
	\For{k = 0,\dots , K}{
		Uniformly sample $x\in\mathcal{X}$\;
		Compute $\mathbf{u}$ and $\mathbf{d}$ such that $u_t \gets \arg\max_u \min_d Q_{\theta_k}(\xi_x^{\mathbf{u,d}}(t),u,d)$ and $d_t\gets \arg \min_d Q_{\theta_k}(\xi_x^{\mathbf{u,d}}(t),u_t,d) $\;
		\For{t= 0,\dots, T}{
			%With probability $\epsilon$ randomly pick $(u_t,d_t)\in \mathcal{U\times D}$, and otherwise
			Store $\psi_t = (\xi_x^{\mathbf{u,d}}(t),u_t,d_t,\xi_x^{\mathbf{u,d}}(t+1))$ in $\Psi$\;
		}		
		\For{j=1,\dots,J}{
			Sample a transition $\psi_j$ in $\Psi$\;
			Set $y_j \gets \min \{c(\xi_x^{\mathbf{u,d}}(j)), \max\{ r(\xi_x^{\mathbf{u,d}}(j)),$ $\gamma \max_u \min_d Q_{\theta_k}(\xi_x^{\mathbf{u,d}}(j),u,d) \} \}$\;\label{alg:bellman}
		}
		$\theta_{k+1} \gets\theta_{k} - \alpha \nabla_\theta \Big( \sum_{j=1}^J \big((y_j - Q_{\theta} (\xi_x^{\mathbf{u,d}}(j),u_j,d_j))^2+\lambda Q_{\theta}(\xi_x^{\mathbf{u,d}}(j),u_j,d_j)\big)\Big)$\;\label{alg:update_gradient}
	}
\end{algorithm}

\section{Experiments}\label{sec:experiments}
In this section, we present experiments on Algorithm~\ref{alg:DQN} and show that our method could learn a (conservative) approximation to the reach-avoid set reliably even with neural network approximation. 

\subsection{2D Experiment for Reach-Avoid Game}\label{sec:2d_experiment}
In this experiment, we compare the reach-avoid set learned by Algorithm~\ref{alg:DQN} with the one learned by tabular Q-learning, where we first grid the continuous state space and then run value iteration \eqref{eq:bellman_backup} over the grid. We treat the reach-avoid set learned by tabular Q-learning as the ground truth solution. We apply Algorithm~\ref{alg:DQN} to learn neural network Q functions with $4$ hidden layers, where each hidden layer has 128 neurons with ReLu activation functions. This neural network architecture is chosen because empirically it provides sufficient model capacity for approximating the true value function.

Let the target set and constraint set be $\mathcal{T}=\{(x,y)\in\mathbb{R}^2: 1-(x^2+y^2)>0\}$ and $\mathcal{C} = \{(x,y)\in\mathbb{R}^2: 1-\frac{(x-2)^2}{1.5^2} - y^2>0 \}$, respectively. We visualize the target set $\mathcal{T}$ and the constraint set $\mathcal{C}$ in Figure \ref{fig:2d_target_constraint}. Consider the following discrete-time double integrator dynamics, with the time constant $\Delta t = 0.02s$:
\begin{figure}[t]
    \centering
    \includegraphics[clip, trim = 5cm 10.5cm 5cm 10.5cm, width =0.3\textwidth]{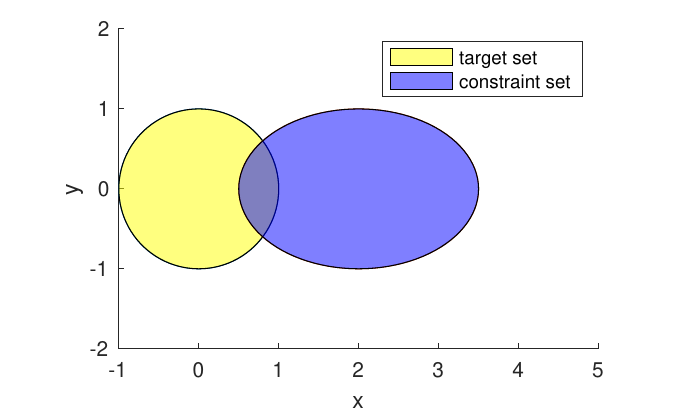}
    \caption{The target set and the constraint set in the 2D experiment}
    \label{fig:2d_target_constraint}
\end{figure}

\begin{figure}[t]
\centering
\includegraphics[clip,trim = 0cm 9.8cm 0cm 9.8cm,width=0.5\textwidth]{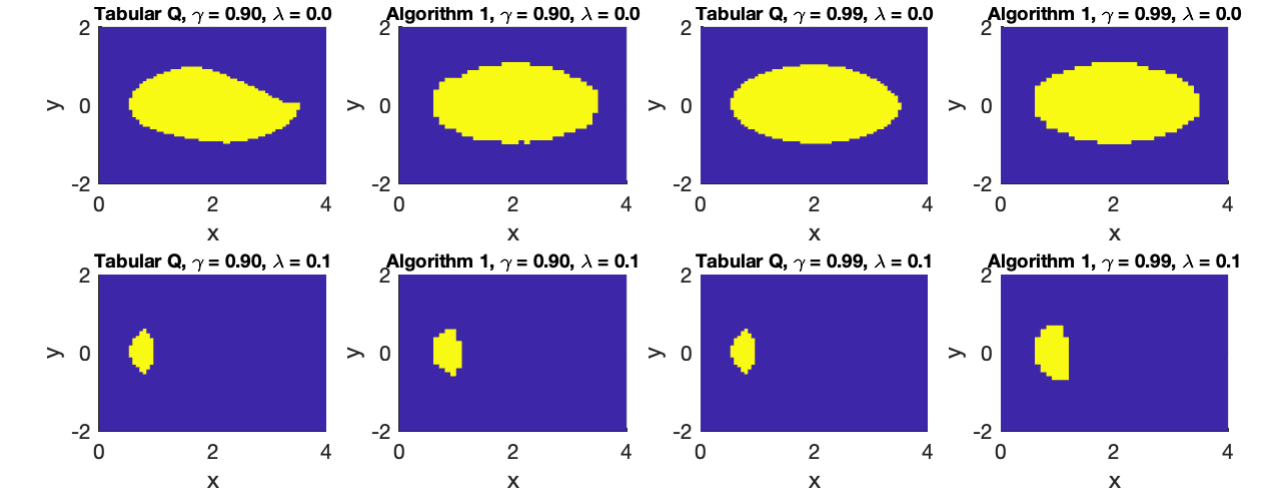}
\caption{Comparison of the 2D reach-avoid set learned by tabular Q-learning and Algorithm~\ref{alg:DQN}. The yellow area corresponds to the reach-avoid set.}
\label{fig:DI}
\end{figure}

\begin{figure}[t]
    \centering
    \includegraphics[clip, trim = 3cm 7.7cm 2.6cm 7.7cm, width =0.35\textwidth]{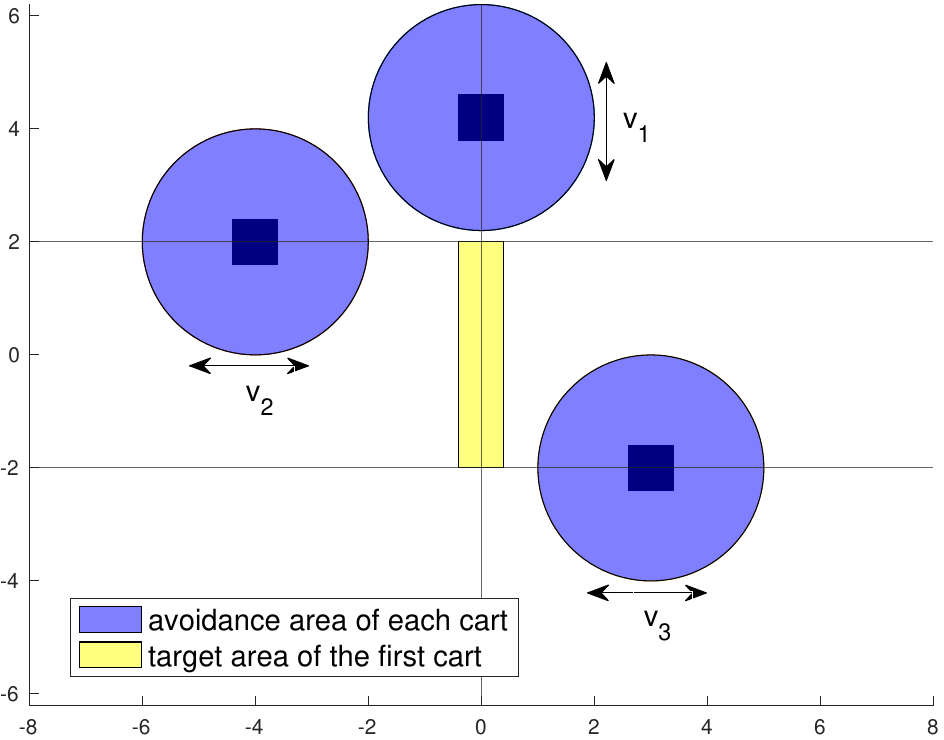}
    \caption{The avoidance area and the target area in the three-cart reach-avoid zero-sum game experiment}
    \label{fig:6d_target_constraint}
\end{figure}

\begin{figure}[h!]
    \centering
    \includegraphics[clip,trim = 1cm 9.2cm 1cm 9cm,width=0.5\textwidth]{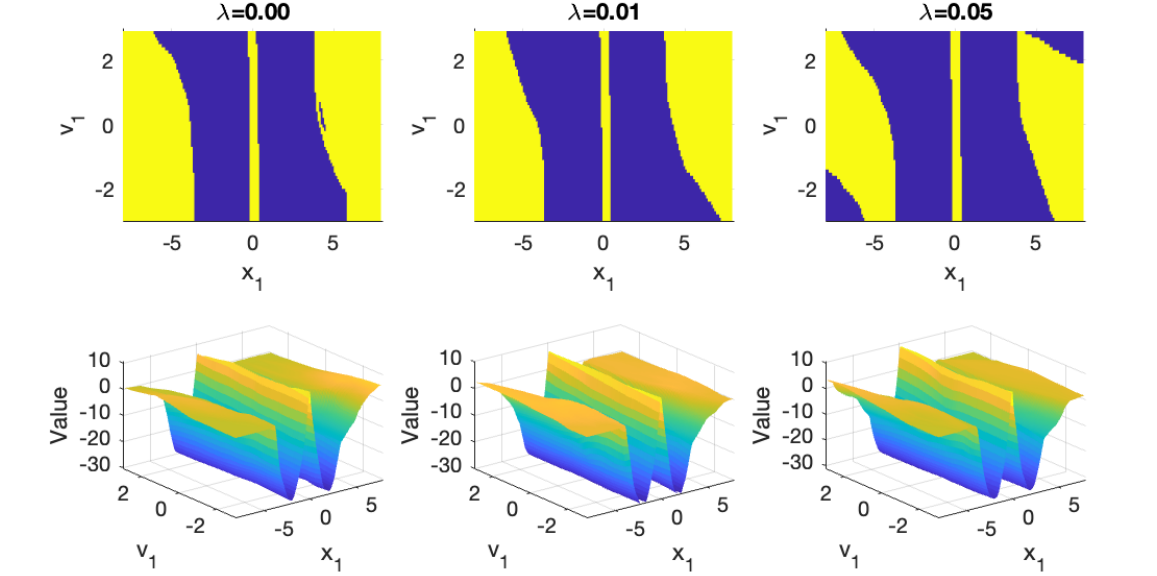}
    \caption{Visualization of reach-avoid sets under different CQL penalty weights $\lambda$, with $[x_2,v_2,x_3,v_3]=[-1,1,1, -1]$. In the first row, the yellow area represents the reach-avoid set. In the second row, we plot the value of each point in the corresponding plots in the first row.}
    \label{fig:DI_3}
\end{figure}

\begin{figure}[h!]
    \centering
    \includegraphics[clip, trim = 1cm 11.4cm 1cm 11.5cm, width = 0.5\textwidth]{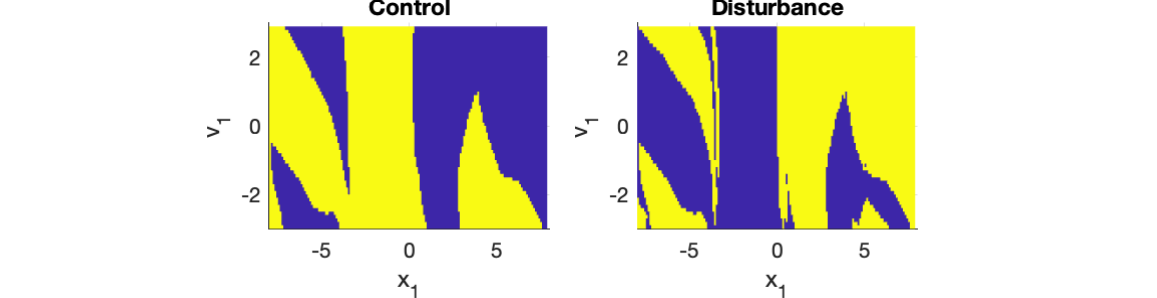}
    \caption{The control and disturbance policies extracted from the neural network value function learned by Algorithm \ref{alg:DQN} with $\lambda = 0.0$. The other four states are $[x_2,v_2,x_3,v_3] = [-1,1,1,-1]$. The yellow and blue areas in the left plot correspond to the control inputs with the values 1 and -1, respectively. The yellow and blue ares in the right plot correspond to the disturbances with the values $0.5$ and $-0.5$, respectively.}
    \label{fig:policy}
\end{figure}
% \begin{figure}[t]
%     \centering
%     \includegraphics[clip, trim=3cm 7cm 3cm 7cm,width =0.3 \textwidth]{plots/just_VK.pdf}
%     \caption{Viability Kernel. The other four states are $x_2 = -1$, $v_2 = 1$, $x_3 = 1$, $v_3 =-1 $.}
%     \label{fig:VK}
% \end{figure}
% \begin{figure}[t]
%     \centering
%     \includegraphics[clip,trim = 1.5cm 11cm 13.5cm 11cm,width=0.3\textwidth]{plots/DI_3_VK_plot.pdf}
%     \caption{Viability Kernel. The other four states are $x_2 = -1$, $v_2 = 1$, $x_3 = 1$, $v_3 =-1 $.}
%     \label{fig:VK}
% \end{figure}

% \begin{figure}[t]
%     \centering
%     \includegraphics[clip,trim = 1cm 11cm 1cm 11cm,width=0.5\textwidth]{plots/DI_3_BRS_plot.pdf}
%     \caption{Backward Reachable Set. The other four states are $x_2 = 0.5$, $v_2 = 0.1$, $x_3 = 0.7$, $v_3 = 0.1$. }
%     \label{fig:BRS}
% \end{figure}
% \begin{figure}[t]
%     \centering
%     \includegraphics[clip, trim=1cm 10cm 1cm 10cm, width =0.4\textwidth]{plots/just_VK_BRS.pdf}
%     \caption{The Viability Kernel is plotted with $x_2 = -1$, $v_2 = 1$, $x_3 = 1$, $v_3 =-1 $. The Backward Reachable Set is plotted with $x_2 = 0.5$, $v_2 = 0.1$, $x_3 = 0.7$, $v_3 = 0.1$.}
%     \label{fig:VK_BRS}
% \end{figure}
\begin{figure}[t]
    \centering
    \includegraphics[clip, trim = 0cm 11.6cm 0cm 11.6cm, width = 0.5\textwidth]{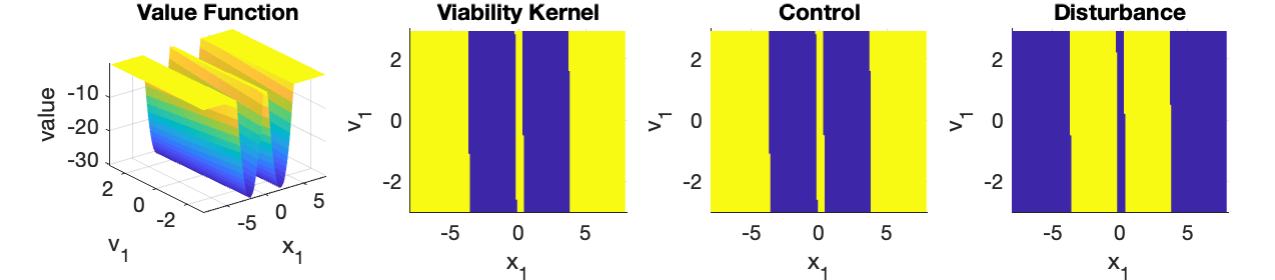}
    \caption{Visualization of the learned viability kernel, with $[x_2,v_2,x_3,v_3]=[1,1,-1,1]$. In the second subplot, the yellow area is the viability kernel. In the third subplot, the yellow and blue areas correspond to the control inputs $1$ and $-1$, respectively. In the fourth subplot, the yellow and blue areas correspond to the disturbances $0.5$ and $-0.5$, respectively.}% Success rate $79.2\%$.% The particular success rate is $74\%$. Training epoch: 2550.}
    \label{fig:VK}
\end{figure}
\begin{figure}[h!]
    \centering
    \includegraphics[clip, trim = 0cm 11.6cm 0cm 11.6cm, width = 0.5\textwidth]{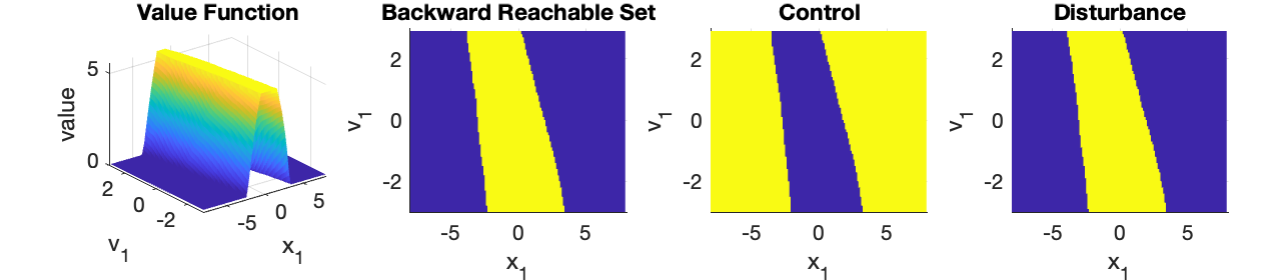}
    \caption{Visualization of the learned backward reachable set, with $[x_2,v_2,x_3,v_3]=[0.6,0.0,0.7,0.1]$. In the second subplot, the yellow area is the backward reachable set. In the third subplot, the yellow and blue areas correspond to the control actions $1$ and $-1$, respectively. In the fourth subplot, the yellow and blue areas correspond to the disturbance action $0.5$ and $-0.5$, respectively.}%The empirical success rate is $73\%$ from 1000 points randomly sampled from the Backward Reachable Set. But for if we fixed the above four axis states and random sampled 1000 points, then the success rate is $100\%$. Training epoch: 2260.}
    \label{fig:BRS}
\end{figure}

\begin{equation}\label{eq:2d_dynamics}
    \begin{bmatrix}
    x(t+1)\\    y(t+1)
    \end{bmatrix} = \begin{bmatrix}
    x(t) + \Delta t\cdot y(t)\\ y(t) + \Delta t\cdot (u(t) + d(t))
    \end{bmatrix}
\end{equation}
where $x(t)$ and $y(t)$ are the position and velocity, respectively. The $u(t)\in\{-1,1\}$ and $d(t)\in\{-0.5,0.5\}$ are the control action and disturbance at time $t\in\{0,1,2,\dots\}$, respectively. %{\color{blue}HJ reachability analysis \cite{bansal2017hamilton} for the continuous-time counterpart of \eqref{eq:2d_dynamics} generates control and disturbance where the sign of them at every state should be the opposite to each other, but this is not necessarily true for discrete-time systems.}%This suggests that the optimal control and disturbance laws for \eqref{eq:2d_dynamics} should be in opposite direction when the products $\Delta t \cdot u$ and $\Delta t\cdot d$ are sufficiently small, but this may not be true when the above two products are large.

We see in Figure~\ref{fig:DI} that the reach-avoid set learned by Algorithm \ref{alg:DQN} is similar to the one computed by tabular Q-learning. An interesting observation is that, under different time-discount factors, the reach-avoid sets computed by tabular Q-learning are somewhat different. This could potentially be due to the numerical difficulty introduced by small values of $\gamma$. To be more specific, from the definition of the value function \eqref{eq:inf_horizon_reach_avoid_problem}, we see that the powers of small constant $\gamma\in[0,1)$ decays to 0 faster than those for a large parameter $\gamma\in[0,1)$, and this makes it difficult to distinguish zero from the product of the power of a small constant $\gamma$ and a bounded function term. Despite this numerical difficulty, it can be observed that both tabular Q-learning and Algorithm \ref{alg:DQN} learn similar reach-avoid sets under different time-discount factors $\gamma\in[0,1)$.

In addition, as we increase the CQL penalty $\lambda$, the reach-avoid set shrinks to a conservative subset of the true reach-avoid set in both tabular Q-learning and Algorithm~\ref{alg:DQN}. For the same $\lambda$, although subject to certain numerical errors, both of these methods yield similar results, which empirically supports Theorem \ref{thm:sandwidch}.

\subsection{6D Experiment for Reach-Avoid Game}\label{sec:6d experiment}
In this experiment, we consider a three-cart dynamical system, where the dynamics of each cart is modeled as a double integrator. The carts move along different axes. The first cart with the position $x_1(t)$ and velocity $v_1(t)$ moves vertically, while the second and the third carts move along the upper and the lower horizontal lines in Figure \ref{fig:6d_target_constraint}, respectively. 

The task is to drive the first cart towards the yellow region in Figure~\ref{fig:6d_target_constraint} while keeping the distance between every two cars at least $2$. The target and the constraint sets can then be formulated as $\mathcal{T} = \{x\in\mathbb{R}^6: 2-|x_1|>0 \}$ and $\mathcal{C} = \{x\in\mathbb{R}^6: \min\{(x_1-2)^2+x_2^2, (x_1+2)^2+x_3^2 \}-4>0\}$, respectively. The dynamics is
\begin{equation}\label{eq:experiment 6d dynamics}
    \begin{aligned}
    \begin{bmatrix}
    x_1(t+1)\\v_1(t+1)\\x_2(t+1)\\v_2(t+1)\\x_3(t+1)\\v_3(t+1)
    \end{bmatrix} = \begin{bmatrix}
    x_1(t) + \Delta t\cdot v_1(t)\\v_1(t) + \Delta t \cdot (u_t+d_t)\\x_2(t) + \Delta t\cdot v_2(t)\\ v_2(t) + 0.02\Delta t\\ x_3(t) + \Delta t\cdot v_3(t)\\v_3(t) +0.02 \Delta t 
    \end{bmatrix}
    \end{aligned}
\end{equation}
where $u_t\in\{-1,1\}$ and $d_t\in \{-0.5,0.5\}$ are the control input and disturbance, respectively. The time-integration constant $\Delta t $ is equal to $0.02 s$. The time-discount factor $\gamma$ is $0.99$.

Due to the curse of dimensionality, tabular Q learning explained in the previous experiment suffers numerical difficulties in this 6-dimensional experiment. In this subsection, we apply Algorithm~\ref{alg:DQN} with the same neural network architecture as in Section~\ref{sec:2d_experiment}. We plot the learned reach-avoid set by projecting it onto a 2-dimensional plane. In Figure~\ref{fig:DI_3}, we visualize the reach-avoid set as well as the effect of the CQL penalty $\lambda$ on the learned reach-avoid set. As the penalty $\lambda$ increases, the volume of the reach-avoid set shrinks while the empirical success rate improves. This suggests that a larger penalty $\lambda$ induces a more conservative estimation of the reach-avoid set and the policy. %In addition, we observe that even with neural network approximations Algorithm~\ref{alg:DQN} may still learn a reasonable reach-avoid set for potentially high-dimensional systems. 

We sample 1000 initial states in each of the three learned reach-avoid set in Figure \ref{fig:DI_3}, and then collect the trajectories with maximum length 1000 time steps, under the control and disturbance policies induced by the learned neural network value function. For the three learned value functions, the portion of those 1000 trajectories under each of the three values of $\lambda$ that successfully reach the target set without violating the constraints are $85.4\%$, $88.4\%$, and $91.3\%$. In addition, we randomly sample 10,000 points in the state space, and observe that the ratios of points lying in the learned reach-avoid sets for $\lambda=0.01$ and $\lambda=0.05$ to those for $\lambda = 0.0$ are $99.7\%$ and $82.5\%$, respectively. However, note that the training epochs taken to learn the three value functions are 1290, 1650, and 1470, respectively.% There is an interplay between the values of the CQL weight $\lambda$ and the training performance of Algorithm \ref{alg:DQN}, but it is beyond the scope of this paper.

In Figure~\ref{fig:policy}, we visualize the control and disturbance policies extracted from the Q function of the neural network corresponding to $\lambda=0.0$ in Figure~\ref{fig:DI_3}. The two policies are considered as to be reasonable because at most state, the control policy either drives the agent towards the target set or prevent the violation of the constraint. %{\color{red}The control and disturbance laws make sense at most states. For example, for the initial condition $x_1=-5$ and $v_1=0$, the control input drives the first cart towards the target set by taking the value $1$ while the disturbance hopes to prevent this happening by picking the vlaue $-0.5$. However, when $v_1=2$, the control input picks the value $-1$ and the disturbance has the value $0.5$, which is because the speed of the first cart is too big and the control input wants to avoid collision with the third cart.} 
Notice that there are a few states, e.g., some points at the right lower corner of the disturbance law in Figure \ref{fig:policy}, where they do not complement to each other. One possible reason is that we learn a local optimal neural network Q function by minimizing the non-convex loss function \eqref{eq:cql_loss}.

\subsection{6D Experiments on Learning Viability Kernel and Backward Reachable Set}
In this subsection, we apply Algorithm~\ref{alg:DQN} to learn viability kernel and backward reachable set for the 6-dimensional dynamical system in Subsection \ref{sec:6d experiment}. The results empirically confirm Propositions~\ref{theorem:pure_safe} and \ref{theorem:pure_reach}. In the following two experiments, the same neural network architecture as in Section~\ref{sec:6d experiment} does not yield satisfactory results and we conjecture that this is due to the limited model capacity. As such, in this subsection, we increase the model capacity by adopting 4-layer neural networks with 256 neurons in each layer. We set the CQL penalty parameter to be $\lambda=0.0$.

We first consider learning the viability kernel where the constraint set is the same as the one in Subsection~\ref{sec:6d experiment}. The reward function is set to be $r(x) = -1$, for all $ x\in\mathbb{R}^n$. %As suggested in Proposition~\ref{theorem:pure_safe}, we set the reward function to be the constant function $r(x)=-1 $ for all $x\in\mathbb{R}^n$. %shows that the viability kernel can be learned by setting the reward function to be a constant function $-1$ and apply Algorithm~\ref{alg:DQN}. With this particular reward function and the constraint function in Subsection \ref{sec:6d experiment}, we apply Algorithm~\ref{alg:DQN} to learn a neural network value function whose zero level set in theory recovers the viability kernel. 
%We learn a neural network Q function by applying Algorithm~\ref{alg:DQN}. 
In the first subplot of Figure~\ref{fig:VK}, we visualize the learned value function. %The maximum of the learned value function is $-0.0082$ and 
The value function is non-positive, which is predicted by Proposition \ref{theorem:pure_reach}. %From the system dynamics, we believe that the viability kernel should be nonempty, which suggests that the maximum value should be zero. This inconsistency could be due to the neural network approximation error. As such, 
%Given the neural networks approximation error, we heuristically classify any value greater than $-0.2$ as a zero value, and 
We visualize the learned viability kernel in the second subplot of Figure \ref{fig:VK}. %The learned viability kernel is reasonable because at most states in it, there exists a heuristic control input driving the first cart away from the other two carts. 
We sampled $1000$ initial positions in the learned viability kernel and simulated a trajectory for each of them with $600$ steps. The portion of those sampled points that can be maintained inside the constraint set is~$79.2\%$.

Subsequently, we consider learning a backward reachable set by leveraging Proposition \ref{theorem:pure_reach}, where the constraint function is set to be a constant function with the value $1$. % and a neural network Q function is learnt by Algorithm~\ref{alg:DQN}. 
%Since there exists a controller which can drive each initial state towards the target set in Subsection \ref{sec:6d experiment}, the backward reachable set of the target set in Subsection~\ref{sec:6d experiment} is the entire state space. To make this problem nontrivial, 
We consider the target set $\mathcal{T}=\{x\in\mathbb{R}^6:\min(2-|x_1|, 1-|x_2|,1-|x_3|)>0\}$. We visualize the learned value function and backward reachable set in Figure~\ref{fig:BRS}. As shown in Proposition \ref{theorem:pure_reach}, the value function is non-negative. %Considering the neural network approximation error, we classify any value less than 0.2 as the nominal zero value. %As plotted in the second subplot in Figure \ref{fig:BRS}, the learned backward reachable set is reasonable because for each initial state therein the initial state is either already in the target set or can be driven to it by taking the control policy in Figure~\ref{fig:BRS}. 
We sampled $1000$ points in the obtained backward reachable set. The portion of initial states that can be driven towards the target set within 600 simulation steps~is~$70.0\%$.

It is observed that the empirical success rate of viability kernel is higher than backward reachable set. There are multiple reasons behind it. First, the simulation horizon is not long enough for an initial state to be driven towards the backward reachable set. In addition, for backward reachable set, a single failure in control may push the trajectory away from the target set in the future. However, for viability kernel, multiple failures in control are allowed if controls at the boundary of the viability kernel keep the trajectory inside the constraint set.

We remark here that the high-dimensional experiments are challenging because there is a lack of efficient methods to check the sufficiency of neural network model capacity and to solve the non-convex optimization problem \eqref{eq:cql_loss}, which are common problems in deep RL and on-going research directions in the deep RL community \cite{fan2020theoretical}.

\section{Conclusion}
\label{sec:conclusion}
In this paper, we investigated the infinite-horizon reach-avoid zero-sum game problem, in which the goal is to learn the reach-avoid set in the state space and an associated policy such that each state in the reach-avoid set could be driven towards a given target set while satisfying constraints. We designed a value function that offers several properties: 1) its super-zero level set coincides with the reach-avoid set and the induced Bellman Backup equation is a contraction mapping; 2) the value function is Lipschitz continuous under certain conditions; and 3) the value function can be adapted to compute the viability kernel and backward reachable set. %In particular, we combined the proposed Hamilton-Jacobi method with a deep Reinforcement Learning technique, the Conservative Q-Learning, to alleviate the curse of dimensionality issue. 
We proposed to alleviate the curse of dimensionality issue by developing a deep RL algorithm. The provided theoretical and empirical results suggest that our method is able to learn reach-avoid sets reliably even with neural network approximation errors.
%Particularly, to compensate the NN approximation error by allowing a conservativeness tuning, we extended conservative Q learning methods. Our theories (...) have been validated by 2D and 6D experiments.

%We designed a value function whose super-zero level sets serves as the reach-avoid set, and extended conservative Q learning method to learn the designed value function, whose super-zero level set serves as a conservative estimation of the true reach-avoid set. 

\section*{Appendix}

% \begin{proof}
% Without loss of generality, there are only two cases. Either $a(c) \le  b$, for $\forall c\in\mathcal{C}$, or $\exists c\in \mathcal{C} $ such that $a(c)> b$. For the first case, $\min (a(c),b) = a(c)$ and $\max_{c\in\mathcal{C}} a(c) \le b$. Thus,
% \begin{equation}
%     \max_{c\in\mathcal{C}} \min(a(c), b) = \max_{c\in\mathcal{C}} a(c) = \min(\max_{c\in\mathcal{C}} a(c), b).
% \end{equation}
% We then consider the second case. Suppose $a(c') >b$, then $\min(a(c'), b) = b$, which implies that $\max_{c\in\mathcal{C}} \min(a(c),b) = b$. Observe that $\min(\max_{c\in\mathcal{C}}a(c) ,b) = b$ as well. Thus, \eqref{eq:max_min_swap_eq} holds true.
% \end{proof}

% \begin{lemma}\label{lem:double_min}
% \begin{equation}
%     \max\big\{ \min\{a,b\} ,\min\{ b,c \} \big\} = \min \big\{b, \max\{ a,c \} \big\}
% \end{equation}
% \end{lemma}
% \begin{proof}
% Without loss of generality, assume $a\le c$. We have three cases: (1) $a\le c\le b$; (2) $b\le a\le c$; (3) $a\le b\le c$.

% For the first case, we have 
% \begin{equation}
%     \max\big\{ \min\{a,b\}, \min\{b,c\} \big\} = c = \min\{b\max\{a,c\}\}
% \end{equation}

% For the second case, we have
% \begin{equation}
%     \max\big\{ \min\{a,b\}, \min\{b,c\} \big\} = b = \min\{b\max\{a,c\}\}
% \end{equation}

% For the last case, we have
% \begin{equation}
%     \max\big\{ \min\{a,b\}, \min\{b,c\} \big\} = b = \min\{b\max\{a,c\}\}
% \end{equation}

% We complete the proof.
% \end{proof}

\begin{proof}[Proof of Theorem~\ref{thm:V>0}]
We first prove the sufficiency. Suppose that there exist an $\epsilon>0$ and a finite $T>0$ such that for all $\phi\in\Phi$, there is a sequence of control inputs $\mathbf{u}=\{u_t\}_{t=0}^\infty$ with the property $
%\begin{equation*}\small
     \min\{\min_{t=0,\dots,T} \gamma^t c(\xi_x^{\mathbf{u},\phi (\mathbf{u})}(t)), \gamma^T r(\xi_{x}^{\mathbf{u},\phi(\mathbf{u})}(T))\}>\epsilon,$ 
%\end{equation*}
which implies that for all $\phi\in\Phi$.
% ,
% \begin{equation*}\small
% \begin{aligned}
%     \max_{\{u_t\}_{t=0}^\infty}\sup_{T=0,1,\dots}& \min\{\min_{t=0,\dots,T} \gamma^t c(\xi_x^{\mathbf{u},\phi (\mathbf{u})}(t)), \gamma^T r(\xi_{x}^{\mathbf{u},\phi(\mathbf{u})}(T))\}>\epsilon.
% \end{aligned}
% \end{equation*}
We then have $V(x) \ge \epsilon >0$.
% \begin{equation*}\small
% \begin{aligned}
%     V(x) = \inf_{\phi} \max_{\{u_t\}_{t=0}^\infty}& \sup_{T=0,1,\dots} \min\{ \min_{t=0,\dots, T} \gamma^t c(\xi_{x}^{\mathbf{u},\phi(\mathbf{u})}(t)), \\ &\gamma^T r(\xi_x^{\mathbf{u},\phi(\mathbf{u})}(T)) \}\ge \epsilon >0.
% \end{aligned}
% \end{equation*}
% where $\gamma^t$ is strictly positive for any finite time $t>0$ and multiplying $\gamma^t$ to a term does not change the sign of that term. 

We show the necessity by a contrapositive statement. Suppose that for all $\epsilon>0$, there exists a disturbance strategy $\phi\in\Phi$ such that for all sequences of control inputs $\{u_t\}_{t=0}^\infty$ and all time $T\in\{0,1,\dots\}$, $
%\begin{equation*}\small
\min\{\min_{t=0,\dots, T}\gamma^t c(\xi_{x_0}^{\mathbf{u},\phi(\mathbf{u})}(t)),\gamma^T r(\xi_{x_0}^{\mathbf{u},\phi(\mathbf{u})}(T))\}\le \epsilon,$ 
%\end{equation*}
which is equivalent to that for all $\epsilon>0$, there exists a disturbance strategy $\phi$ with the property
\begin{equation*}\small
\begin{aligned}
    \max_{\{u_t\}_{t=0}^\infty} \sup_{t=0,1,\dots}&\min\big\{ \min_{t=0,\dots, T} \gamma^t c(\xi_{x_0}^{\mathbf{u},\phi(\mathbf{u})}(t)) , \gamma^T r(\xi_{x_0}^{\mathbf{u},\phi(\mathbf{u})}(T))\big\}\le \epsilon.
\end{aligned}
\end{equation*}
In another words, for all $ \epsilon>0$,
\begin{equation*}\small
\begin{aligned}
    V(x_0) = \inf_{\phi} \max_{\{u_t\}_{t=0}^\infty}& \sup_{t=0,1,\dots} \min\big\{ \min_{t=0,\dots, T}\gamma^t c(\xi_{x_0}^{\mathbf{u},\phi(\mathbf{u})}(t)), \\&\gamma^T r(\xi_{x_0}^{\mathbf{u},\phi(\mathbf{u})}(T)) \big\}\le \epsilon,
\end{aligned}
\end{equation*}
which yields that
\begin{equation*}\small
\begin{aligned}
    V(x_0) = \inf_{\phi} \max_{\{u_t\}_{t=0}^\infty}& \sup_{t=0,1,\dots} \min\big\{ \min_{t=0,\dots, T}\gamma^t c(\xi_{x_0}^{\mathbf{u},\phi(\mathbf{u})}(t)), \\ &\gamma^T r(\xi_{x_0}^{\mathbf{u},\phi(\mathbf{u})}(T)) \big\}\le0.
\end{aligned}
\end{equation*}

Therefore, $V(x)>0$ if and only if the problem \eqref{eq:inf_horizon_reach_avoid_problem} is finite-time feasible and reachable for the initial state $x$.
\end{proof}
%Before we present the proof of Theorem~\ref{thm:Bellman}, we first introduce Lemma~\ref{lem: max_min_swap}.
\begin{lemma}\label{lem: max_min_swap} Given a compact set $\mathcal{C}\subseteq \mathbb{R}$ and a function $a(\cdot):\mathbb{R}\to \mathbb{R}$, we have $\max_{c\in \mathcal{C} }  \min (a(c), b) = \min (\max_{c \in \mathcal{C} } a(c), b),\label{eq:max_min_swap_eq} $ and $\min_{c\in \mathcal{C} }  \min (a(c), b) = \min (\min_{c \in \mathcal{C} } a(c), b) $.
%\begin{equation}
%\begin{aligned}
%    \max_{c\in \mathcal{C} }  \min (a(c), b) &= \min (\max_{c \in \mathcal{C} } a(c), b),\label{eq:max_min_swap_eq}\\
%    \min_{c\in \mathcal{C} }  \min (a(c), b) &= \min (\min_{c \in \mathcal{C} } a(c), b).
%\end{aligned}
%\end{equation}
\end{lemma}
\begin{proof}[Proof of Theorem~\ref{thm:Bellman}]
We first show that $V(x)$ as defined in \eqref{eq:inf_horizon_reach_avoid_problem} satisfies the Bellman backup equation \eqref{eq:bellman_backup}, and then show that the proposed Bellman backup \eqref{eq:bellman_backup} is a contraction mapping.

By definition, we have
\begin{equation*}\small
\begin{aligned}
V(x_0) =& %\inf_{\phi}\max_{\{u_t\}_{t=0}^\infty} \sup_{t=0,1,\dots} \min\Big\{ \gamma^t r(\xi_{x_0}^{\mathbf{u},\phi(\mathbf{u})}(t)), \\&\min_{\tau=0,\dots,t} \gamma^{\tau} c(\xi_{x_0}^{\mathbf{u},\phi(\mathbf{u})}(\tau)) \Big\}\\
%= & \inf_{\phi} \max_{\{u_t\}_{t=0}^\infty}\max\Big\{\min\{r(x_0),c(x_0)\} , \sup_{t=1,2,\dots}\\&  \min\{\gamma^t r(\xi_{x_0}^{\mathbf{u},\phi(\mathbf{u})}(t)), \min_{\tau=0,\dots,t} \gamma^{\tau} c(\xi_{x_0}^{\mathbf{u},\phi(\mathbf{u})}(\tau))\} \Big\}\\
%= &\max \Big\{ \min\{r(x_0), c(x_0)\}, \inf_{\phi}\max_{\{u_t\}_{t=0}^\infty} \sup_{t=1,2,\dots, } \\&\min\{ \gamma^t r(\xi_{x_0}^{\mathbf{u},\phi(\mathbf{u})}(t)),\min_{\tau=0,\dots, t} \gamma^{\tau}c(\xi_{x_0}^{\mathbf{u},\phi(\mathbf{u})}(\tau)) \}  \Big\}\\
\max\Big\{ \min\{ r(x_0), c(x_0) \},  \inf_{\phi}\max_{\{u_t\}_{t=0}^\infty} \sup_{t=1,2,\dots} \min \{ \min\{\\&\gamma^t r(\xi_{x_0}^{\mathbf{u},\phi(\mathbf{u})}(t)), \min_{\tau=1,\dots, t} \gamma^{\tau} c(\xi_{x_0}^{\mathbf{u},\phi(\mathbf{u})}(\tau))\}, c(x_0) \}\Big\} .
\end{aligned}
\end{equation*}
% We continue the above derivation as,
By applying Lemma \ref{lem: max_min_swap}, the above equation implies
\begin{equation*}\small
\begin{aligned}
V(x_0)%=&\max\Big\{ \min\{r(x_0),c(x_0)\}, \max_{u_0}\min_{d_0}\max_{u_1}\min_{d_1}\dots \max_{t=1,\dots,\infty} \min\{ \\& \min\{ \gamma^t r(\xi_{x_0}^{\mathbf{u,d}}(t)),\min_{t'=1,\dots,t}\gamma^{t'} c(\xi_{x_0}^{\mathbf{u,d}}(t')) \},c(x_0) \} \Big\} \\
% = &\max\Big\{ \min\{r(x_0),c(x_0)\},  \min\{ \inf_{\phi}\max_{\{u_t\}_{t=0}^\infty} \sup_{t=1,2,\dots} \\&\min\{\gamma^t r(\xi_{x_0}^{\mathbf{u},\phi(\mathbf{u})}(t)), \min_{\tau=1,\dots,t}\gamma^{\tau} c(\xi_{x_0}^{\mathbf{u},\phi(\mathbf{u})}(\tau))\} , c(x_0)\} \Big\} \\
% = &\max\Big\{ \min\{ r(x_0),c(x_0) \}, \min \big\{\gamma \inf_{\phi}\max_{\{u_t\}_{t=0}^\infty} \sup_{\tilde{t} = 0,1,\dots}\min \{\\&   \gamma^{\tilde{t}} r(\xi_{x_0}^{\mathbf{u},\phi(\mathbf{u})}(\tilde{t}+1)), \min_{\tilde{t}' = 0,\dots, t} \gamma^{\tilde{t}'}c(\xi^{\mathbf{u},\phi(\mathbf{u})}_{x_0}(\tilde{t}'+1)) \}, c(x_0) \big\} \Big\}\\
% =&\max\Big\{ \min\{ r(x_0),c(x_0) \}, \min \big\{\gamma \max_{u_0}\min_{d_0}\inf_{\phi}\max_{\{u_t\}_{t=1}^\infty} \sup_{\tilde{t} = 0,1,\dots}\\&  \min \{ \gamma^{\tilde{t}} r(\xi_{x_0}^{\mathbf{u},\phi(\mathbf{u})}(\tilde{t}+1)), \min_{\tilde{t}' = 0,\dots, t}\\& \gamma^{\tilde{t}'}c(\xi^{\mathbf{u},\phi(\mathbf{u})}_{x_0}(\tilde{t}'+1)) \}, c(x_0) \big\} \Big\} \\
% =& \max\Big\{ \min\{ r(x_0),c(x_0) \}, \min \{c(x_0) ,\\& \gamma \max_{u_0}\min_{d_0} V(f(x_0,u_0,d_0))\} \Big\}\\
= &\min\Big\{ c(x_0), \max\{ r(x_0) , \gamma \max_{u_0}\min_{d_0} V(f(x_0,u_0,d_0)) \} \Big\}.
\end{aligned}
\end{equation*}
% where the first equality results from Lemma~\ref{lem: max_min_swap}, the third equality follows from \eqref{eq:max_min_swap_eq} and the last equality follows from the observation that for all $ a,b,c\in\mathbb{R}$,
% $\max\big\{ \min\{a,b\} ,\min\{ b,c \} \big\} = \min \big\{b, \max\{ a,c \} \big\}$.
%Lemma~\ref{lem:double_min}.
%Let $\pi^*$ be the optimal policy. Without loss of generality, at the time $t=0$, we have four cases: (1) $r(s_0)> \gamma \max_{a_0} V(s_1)$ and $c(s_0)\le r(s_0)$; (2) $ r(s_0) > \gamma \max_{a_0} V(s_1) $ and $ c(s_0)> r(s_0) $; (3) $ r(s_0)\le \gamma \max_{a_0} V(s_1) $ and $c(s_0) \le \gamma \max_{a_0} V(s_1) $; (4) $ r(s_0)\le \gamma \max_{a_0} V(s_1) $ and $ c(s_0)> \gamma \max_{a_0} V(s_1)$.
%For the first case, we have 
%\begin{equation}
%    \begin{aligned}
%    r(s_0) &> \gamma \max_{a_0}V(s_1)\\
%    & = \left.\gamma \max_{t=1,\dots,\infty}\min \{ \gamma^t r(s_t), \min_{t'=1,\dots, t} \gamma^{t'}c(s_{t'}) \}\right|_{\pi^*}\\
%    & = 
%    \end{aligned}
%\end{equation}
\noindent In what follows, we show that \eqref{eq:bellman_backup} is a contraction mapping. %Let $V^{(0)}(\cdot):\mathbb{R}^n\to \mathbb{R}^n$ be an arbitrary bounded function. For each non-negative integer $k$, we define 
%\begin{equation*}
%V^{(k+1)}(x): = \min \{ c(x), \max\{ r(x),\gamma \max_u \min_d V^{(k)}(f(x,u,d)) \} \}. 
%\end{equation*}
It suffices to show that for all state $x$, we have $| B[V_1](x) - B[V_2](x) |\le \gamma \left\| V_1 - V_2 \right\|_{\infty}$, where $V_1$ and $V_2$ are two arbitrary bounded functions. One can write \begin{equation}\label{eq:proof_Bellman_backup1}\small
    \begin{aligned}
    &|B[V_1] (x) - B[V_2](x)|\\%& = |\min\{c(x),\max\{r(x),\gamma \max_u\min_d V_k(f(x,u,d))\}\} \\&\ \ -\min\{c(x),\max\{r(x), \gamma \max_u \min_d V_{k-1}(f(x,u,d))\}\}| \\
    &\le | \max\{ r(x), \gamma \max_u\min_d V_1(f(x,u,d)) \}\\&\  \ - \max\{ r(x),\gamma \max_{u}\min_d V_2 (f(x,u,d)) \}|\\
    &\le | \gamma \max_u \min_d V_1(f(x,u,d))  - \gamma \max_u \min_d V_2(f(x,u,d))|
    \end{aligned}
\end{equation}
where the above inequalities follow from the fact that $\|\min\{a,b\} - \min\{a,c\}\|\le b-c, \forall a,b,c\in\mathbb{R}$.
% \begin{equation}
%     \begin{aligned}
%     \|\min\{a,b\} - \min\{a,c\}\|\le b-c, \forall a,b,c\in\mathbb{R}
%     \end{aligned}
% \end{equation}
%Lemma~\ref{lemma:min_max_inequality}.
Without loss of generality, we suppose that $\max_u\min_d V_1(f(x,u,d))\ge \max_u\min_d V_2(f(x,u,d))$ and let $u^*:=\argmax_u \min_d V_1(f(x,u,d))$. In addition, let $d^*:=\arg\min_d V_2(f(x,u^*,d))$. We then have
\begin{equation}\label{eq:proof_Bellman_backup2}\small
    \begin{aligned}
    |\gamma &\max_u\min_d V_1(f(x,u,d))-\gamma \max_u\min_d V_2 (f(x,u,d))|\\&\le \gamma |\min_d V_1( f(x,u^*,d) ) - \min_d V_2(f(x,u^*,d))|\\
    % & \le \gamma \max_u |\min_d V_k(f(x,u,d)) - \min_d V_{k-1}(f(x,u,d))| \\
    & \le \gamma | V_1(f(x,u^*,d^*)) - V_2(f(x,u^*,d^*)) |\\
    & \le \gamma \max_{x} |V_1(x) - V_2(x)|
    \le \gamma \|V_1 - V_2\|_{\infty}.
    \end{aligned}
\end{equation}
Combining \eqref{eq:proof_Bellman_backup1} and \eqref{eq:proof_Bellman_backup2}, we have $|B[V_1](x) - B[V_2](x)|\le \gamma \| V_1 - V_2\|_\infty,\ \textrm{for all }x$. 
%\begin{equation*}
%    \begin{aligned}
%    |B[V_1](x) - B[V_2](x)|\le \gamma \| V_1 - V_2\|_\infty,\ \textrm{for all }x,
%    \end{aligned}
%\end{equation*}
This yields that
$\| B[V_1] - B[V_2] \|_\infty\le \gamma \| V_1 - V_2 \|_\infty$.
\end{proof}

\begin{proof}[Proof of Theorem~\ref{lemma:continuity}]
Define 
\begin{align}\small
    P(\mathbf{u},\mathbf{d},t,x) \coloneqq \min\{ \gamma^t r(\xi_x^{\mathbf{u},\mathbf{d}}(t)), \min_{\tau=0,...,t} \gamma^\tau c(\xi_x^{\mathbf{u},\mathbf{d}}(\tau)) \}.
\end{align}

Consider two initial states $x_1,x_2\in\mathbb{R}^n$.
Given $\epsilon>0$, there exists a non-anticipative strategy $\bar{\phi}$ such that
\begin{align}\small
    V(x_1) \geq \sup_\mathbf{u} \sup_t P(\mathbf{u},\bar{\phi}(\mathbf{u}), t,x_1) - \epsilon.
    \label{eq:pf_lma_cont_eq1}
\end{align}
We select a control sequence $\bar{\mathbf{u}}$ and time $\bar{t}$ such that
\begin{align}\small
    V(x_2) \leq  P(\bar{\mathbf{u}},\bar{\phi}(\bar{\mathbf{u}}), \bar{t},x_2) + \epsilon.
    \label{eq:pf_lma_cont_eq2}
\end{align}
Moreover, \eqref{eq:pf_lma_cont_eq1} implies that
\begin{align}\small
    V(x_1) \geq  P(\bar{\mathbf{u}},\bar{\phi}(\bar{\mathbf{u}}), \bar{t},x_1) - \epsilon.
    \label{eq:pf_lma_cont_eq3}
\end{align}
By combining \eqref{eq:pf_lma_cont_eq2} and \eqref{eq:pf_lma_cont_eq3}, we obtain
\begin{equation}\small
\begin{aligned}
    V(x_1)  - V(x_2 ) + &2\epsilon \geq P(\bar{\mathbf{u}},\bar{\phi}(\bar{\mathbf{u}}), \bar{t},x_1) - P(\bar{\mathbf{u}},\bar{\phi}(\bar{\mathbf{u}}), \bar{t},x_2) \\
    % = & \min\{ \gamma^{\bar{t}} r(\xi_{x_1}^{\bar{\mathbf{u}},\bar{\phi}(\bar{\mathbf{u}})}(\bar{t})), \min_{\tau=0,...,\bar{t}} \gamma^\tau c(\xi_{x_1}^{\bar{\mathbf{u}},\bar{\phi}(\bar{\mathbf{u}})}(\bar{t})) \}\\
    % -& \min\{ \gamma^{\bar{t}} r(\xi_{x_2}^{\bar{\mathbf{u}},\bar{\phi}(\bar{\mathbf{u}})}(\bar{t})), \min_{\tau=0,...,\bar{t}} \gamma^\tau c(\xi_{x_2}^{\bar{\mathbf{u}},\bar{\phi}(\bar{\mathbf{u}})}(\bar{t})) \}\\
    \geq & \min\{ \gamma^{\bar{t}} (r(\xi_{x_1}^{\bar{\mathbf{u}},\bar{\phi}(\bar{\mathbf{u}})}(\bar{t}))-r(\xi_{x_2}^{\bar{\mathbf{u}},\bar{\phi}(\bar{\mathbf{u}})}(\bar{t}))),\\
    &\min_{\tau=0,...,\bar{t}} \gamma^\tau (c(\xi_{x_1}^{\bar{\mathbf{u}},\bar{\phi}(\bar{\mathbf{u}})}(\bar{t})) - c(\xi_{x_2}^{\bar{\mathbf{u}},\bar{\phi}(\bar{\mathbf{u}})}(\bar{t}))) \}.
\end{aligned}
\label{eq:pf_lma_cont_eq4}
\end{equation}
The last inequality is due to the fact that for any finite values $a_i,b_i\in \mathbb{R}$ ($i=1,...,n$), we have
\begin{align}\small
\begin{split}
    \min\{a_1,...,a_n\}&-\min\{b_1,...,b_n\} \geq \min\{a_1-b_1,...,a_n-b_n\}.
\end{split}
\label{eq:pf_lma_cont_eq5}
\end{align}
According to the Lipscthiz-continuity assumptions, 
\begin{align}\small
\begin{split}
    & \|\xi_{x_1}^{\mathbf{u},\mathbf{d}}(t) - \xi_{x_2}^{\mathbf{u},\mathbf{d}}(t)  \| \leq L_f^t\|x_1-x_2\|, \\
    & \|c(\xi_{x_1}^{\mathbf{u},\mathbf{d}}(t)) - c(\xi_{x_2}^{\mathbf{u},\mathbf{d}}(t))  \| \leq L_c L_f^t\|x_1-x_2\|, \\
    & \|r(\xi_{x_1}^{\mathbf{u},\mathbf{d}}(t)) - r(\xi_{x_2}^{\mathbf{u},\mathbf{d}}(t))  \| \leq L_r L_f^t\|x_1-x_2\|.
\end{split}
\end{align}
Thus, it follows from \eqref{eq:pf_lma_cont_eq4} that
\begin{align}\small
\begin{split}
    V&(x_1)  - V(x_2 ) + 2\epsilon  \\
    &\geq-\max\{ L_r \gamma^{\bar{t}}L_f^{\bar{t}}, \max_{\tau=0,...,\bar{t}} L_c \gamma^{\tau}L_f^{\tau } \} \|x_1-x_2\|.
\end{split}
\label{eq:pf_lma_cont_eq55}
\end{align}
The condition $\gamma L_f <1$ implies that the coefficient in the right-hand side in \eqref{eq:pf_lma_cont_eq55} is bounded for all $\bar{t}$. As a result, $   V(x_1)  - V(x_2 ) + 2\epsilon \geq -C \|x_1-x_2\|
%    \label{eq:pf_lma_cont_eq6}
$
for some $C>0$. Similarly, we can show that $
%\begin{align}\small
    V(x_2)  - V(x_1 ) + 2\epsilon \geq -C \|x_1-x_2\|.
%    \label{eq:pf_lma_cont_eq7}
%\end{align}
$
Combining these two inequalities, we prove Theorem \ref{lemma:continuity}.
\end{proof}

\begin{proof}[Proof of Proposition~\ref{theorem:pure_safe}]
The non-positivity of $V(x)$ is due to $r(x)=-1$, for all $x$. We then show that $V(x)=0$ if and only if $x\in\Omega (\bar{\mathcal{C}})$.

(Sufficiency) On one hand, suppose that for all $\phi\in\Phi$, there exists a sequence of control inputs $\mathbf{u}=\{u_t\}_{t=0}^\infty$ under which $c(\xi_{x}^{\mathbf{u},\phi(\mathbf{u})}(t))\ge 0$, $\forall t\in\{0,1,\dots\}$. Then, for all $\phi\in\Phi$,\vspace{-0.15cm}
\begin{equation*}\small
    \begin{aligned}
    V(x)& \ge \max_{\{u_t\}_{t=0}^\infty} \sup_{t=0,1,\dots} \min\big\{ -\gamma^t , \min_{\tau = 0,\dots, t} \gamma^\tau c(\xi_x^{\mathbf{u},\phi(\mathbf{u})}(\tau)) \big\}\\
    & \ge \max_{\{u_t\}_{t=0}^\infty} \min\big\{ -\gamma^t, \min_{\tau = 0,\dots,t}\gamma^\tau c(\xi_x^{\mathbf{u},\phi(\mathbf{u})}(\tau)) \big\}\big|_{t=\infty}=0.
    \end{aligned}\vspace{-0.15cm}
\end{equation*}
On the other hand, \vspace{-0.15cm}
\begin{equation*}\small
    \begin{aligned}
    V(x)&\le \inf_{ \phi } \max_{\{u_t\}_{t=0}^\infty}\sup_{t=0,1,\dots} \min\big\{ 0, \min_{\tau=0,\dots,t}\gamma ^{\tau}c(\xi_{x}^{\mathbf{u},\phi(\mathbf{u})}(\tau)) \big\}\\
    & \le \inf_\phi \max_{\{u_t\}_{t=0}^\infty }\sup_{t=0,1,\dots} 0= 0.
    \end{aligned}\vspace{-0.15cm}
\end{equation*}
Thus, $V(x)=0$ if for all $\phi\in\Phi$ there exists a sequence of control inputs $\{u_t\}_{t=0}^\infty$ such that $ c(\xi_x^{\mathbf{u},\phi(\mathbf{u})}(t))\ge 0$, $\forall t\in\{0,1,\dots\}$.\\
(Necessity) We approach the proof by a contrapositive statement. Suppose that for every $\phi\in\Phi$ and every sequence of control inputs $\mathbf{u}$, there exists a non-negative integer $t_{\phi,\mathbf{u}}$ such that $c(\xi_x^{\mathbf{u},\phi(\mathbf{u})}(t_{\phi,\mathbf{u}}))< 0$. Then, we have\vspace{-0.15cm}
\begin{equation*}\small
    \begin{aligned}
    V(x) %&= \inf_{\phi} \max_{\mathbf{u}} \sup_{t=0,1,\dots} \min \big\{-\gamma^t, \min_{\tau = 0,\dots, t} \gamma^\tau c(\xi_x^{\mathbf{u},\phi(\mathbf{u})}(\tau))\big\}\\
    & \le \max_{\mathbf{u}} \sup_{t=t_{\phi,\mathbf{u}},t_{\phi,\mathbf{u}}+1,\dots} \min\{-\gamma^t , \gamma^{t_{\phi,\mathbf{u}}} c(\xi_x^{\mathbf{u},\phi(\mathbf{u})}(t_{\phi,\mathbf{u}}))\}\\
    & = \max_{\mathbf{u}} \gamma^{t_{\phi,\mathbf{u}}} c(\xi_x^{\mathbf{u},\phi(\mathbf{u})}(t_{\phi,\mathbf{u}}))<0.
    \end{aligned}\vspace{-0.15cm}
\end{equation*}
This completes the proof.
\end{proof}
\begin{proof}[Proof of Proposition~\ref{theorem:pure_reach}]
The non-negativity of $V(x)$ is due to $c(x)=1$, for all $x$. Subsequently, it suffices to show that $V(x)>0$ if and only if there exists an $\epsilon>0$ such that, for all $\phi\in\Phi$, there exists a sequence of control inputs that renders $\gamma^t \xi_x^{\mathbf{u},\phi(\mathbf{u})}(t)>\epsilon$ and $\gamma^t>\epsilon$, for some non-negative integer $t$.

(Sufficiency) Suppose that there exists $\epsilon>0$ such that for all $\phi\in\Phi$ there exists a sequence of control inputs $\mathbf{u} = \{u_t\}_{t=0}^\infty$, $\gamma^{T_{\phi,\mathbf{u}}} r(\xi_x^{\mathbf{u},\phi(\mathbf{u})}(T_{\phi,\mathbf{u}}))>\epsilon$, for some $T_{\phi,\mathbf{u}}\in\mathbb{Z}_+$. Then, we have\vspace{-0.15cm}
\begin{equation*}\small
    \begin{aligned}
    V(x)& = \inf_\phi \max_{\{u_t\}_{t=0}^\infty} \sup_{t=0,1,\dots} \min\big\{ \gamma^t r(\xi_x^{\mathbf{u},\phi(\mathbf{u})}(t)),\gamma^t \big\}\\
    & \ge \inf_{\phi} \max_{\{u_t\}_{t=0}^\infty} \epsilon\ge \epsilon.
    \end{aligned}\vspace{-0.15cm}
\end{equation*}
(Necessity) We approach the proof by a contrapositive statement. Suppose that for all $\epsilon>0$ there exist $\hat{\phi}$ such that for every sequence of control inputs $\mathbf{u} = \{u_t\}_{t=0}^\infty$, it holds that $\gamma^t r(\xi_{x}^{\mathbf{u},\hat{\phi}(\mathbf{u})}(t))\le \epsilon$ or $\gamma^t\le \epsilon$, $\forall t\in\mathbb{Z}_+$. Then, we have\vspace{-0.15cm}
\begin{equation*}\small
    \begin{aligned}
    V(x) & \le \max_{\{u_t\}_{t=0}^\infty} \sup_{t=0,1,\dots} \min\big\{ \gamma^t r(\xi_x^{\mathbf{u},\hat{\phi}(\mathbf{u})}(t)),\gamma^t \big\}\\
    & \le \sup_{\{u_t\}_{t=0}^\infty} \sup_{t=0,1,\dots}\epsilon\le \epsilon.
    \end{aligned}\vspace{-0.15cm}
\end{equation*}
Since the above ienquality holds true for all $\epsilon>0$, we have $V(x)\le 0$. This completes the proof.
\end{proof}

\begin{lemma}\label{lem:inequality_minimax}
Consider two functions $g(\cdot):\mathbb{R}\to\mathbb{R}$ and $\bar{g}(\cdot):\mathbb{R}\to\mathbb{R}$. Suppose that $g(a,b)\le \bar{g}(a,b)$, $\forall a,b$. Then, we have $\max_a\min_b g(a,b) \le \max_a \min_b \bar{g}(a,b)$.
%\begin{equation}
%    \max_a\min_b g(a,b) \le \max_a \min_b \bar{g}(a,b).
%\end{equation}
\end{lemma}
%\begin{proof}
%Define the functions $h(a):=\min_b g(a,b)$ and $\bar{h}(a) :=\min_b \bar{g}(a,b) $. Let $\bar{b}^*(a)\in\arg\min_b \bar{g}(a,b)$. One can write
%\begin{equation}
%\begin{aligned}
%    h(a)&=\min_b g(a,b)\le g(a,\bar{b}^*(a))\\ &\le \bar{g}(a,\bar{b}^*(a))= \min_b\bar{g}(a,b) = \bar{h}(a)
%\end{aligned}
%\end{equation}
%Then, we have
%$\max_a h(a)\le \max_a \bar{h}(a).$\end{proof}

\begin{proof}[Proof of Theorem~\ref{thm:sandwidch}]
Let $V^{(0)}(\cdot):\mathbb{R}^n\to \mathbb{R}$ be an arbitrary bounded function. Consider the value iteration:
\begin{equation}\small
\begin{aligned}
    V^{(k+1)}(x) = B[V^{(k)}](x),\  \forall k\in\mathbb{Z}_+.
\end{aligned}
\end{equation}
Let $V^{(0)}_{CQL}:\mathbb{R}^n\to \mathbb{R}$, and consider the CQL-based value iteration:
\begin{equation}\small
    V_{CQL}^{(k+1)}(x) = B_{CQL}[V_{CQL}^{(k)}](x),\ \forall k\in\mathbb{Z}_+.
\end{equation}
%\textcolor{red}{$V^{(k)}$ denotes the cost-to-go function produced at iteration $k$.}
%Denote by $V_{CQL}^{(k)}$ the $k$-th iteration of the CQL Bellman backup \eqref{eq:CQL_Bellman_backup}. 
Suppose that $V_{CQL}^{(0)}(x) =V^{(0)}(x),\forall x\in\mathbb{R}^n$. It suffices to show that\vspace{-0.3cm}
\begin{equation}\small
    V^{(k)}(x)-\sum_{i=1}^k \gamma^i \lambda \le V_{CQL}^{(k)}(x) \le V^{(k)}(x)-\lambda,\ \forall x\in\mathbb{R}^n,
\end{equation}
because by taking the limit $k\to \infty$, we have\vspace{-0.2cm}
\begin{equation}\small
    \lim_{k\to \infty} V(x) - \frac{\lambda}{1-\gamma}\le V_{CQL}(x)\le V(x)-\lambda, \forall x\in\mathbb{R}^n.
\end{equation}
To this end, we first show the upper bound $V_{CQL}^{(k)}(x)\le V^{(k)}(x) - \lambda$ by induction. From \eqref{eq:CQL_Bellman_backup}, we have\vspace{-0.2cm}
\begin{equation}\label{eq:first_induction_cql_proof}\small
    V_{CQL}^{(1)}(x) = B[V^{(0)}](x)-\lambda\le V^{(1)}(x)-\lambda,\forall x\in\mathbb{R}^n.
\end{equation}
Therefore,\vspace{-0.1cm}
\begin{equation*}\small
\begin{aligned}
    V_{CQL}^{(2)}(x) \le &\min\big\{ c(x), \max\{r(x), \gamma \max_u\min_d V^{(1)}(f(x,u,d))\} \big\}-\lambda\\
    \le&V^{(2)}(x)-\lambda,
\end{aligned}
\end{equation*}
where we substitute \eqref{eq:first_induction_cql_proof} into the CQL Bellman backup \eqref{eq:CQL_Bellman_backup} of the $V_{CQL}^{(2)}(x)$ and also use the inequality in Lemma \ref{lem:inequality_minimax}. By induction, we have $V_{CQL}^{(k)}\le V^{(k)}-\lambda,\forall x\in\mathbb{R}^n.$ Similarly, we show the lower bound $V_{CQL}^{(k)}(x)\ge V^{(k)}(x) - \sum_{i=1}^k \gamma^i \lambda$ as follows. Since $V_{CQL}^{(1)}(x) \ge V^{(1)}-\lambda$, we have\vspace{-0.2cm}
\begin{equation*}\small
\begin{aligned}
    V_{CQL}^{(2)}(x)\ge& \min\big\{ c(x),\max\{r(x), \\ &\gamma \max_u \min_d (V^{(1)}(f(x,u,d))-\lambda)\} \big\}-\lambda\\
    \ge & \min\big\{c(x)-\gamma \lambda, \max \{ r(x) - \gamma \lambda, \\ &\gamma \max_u \min_d V^{(1)}(f(x,u,d)) - \gamma \lambda \}\big\}-\lambda\\
    =&V^{(2)}(x)-\gamma \lambda - \lambda.
\end{aligned}
\end{equation*}\vspace{-0.1cm}
By induction, we have $V^{(k)}-\sum_{i=1}^k \gamma^i \lambda \le V_{CQL}^{(k)}(x)$.
\end{proof}

% \begin{equation}
%     \begin{aligned}
%     V(s) =\max_\pi \max_{t=0,\dots,\infty}\min\{r(\xi_{s}^\pi(t)), \min_{\tau=0,\dots,t}c(\xi_{s}^\pi)\}
%     \end{aligned}
% \end{equation}
% \begin{equation}
%     \begin{aligned}
%     B[V](s) = \min\{c(s), \max\{r(s), \max_a V(f(s,a))\}\}
%     \end{aligned}
% \end{equation}
% \begin{equation}
%     \begin{aligned}
%     B[V](s) = (1-\gamma)\min\{r(s),c(s)\}+\gamma \min\{ \max\{ \max_a V(f(s,a)), r(s) \}, c(s) \}
%     \end{aligned}
% \end{equation}

% \begin{equation}
% \begin{aligned}
%     V(s) &= \max_u r(s) + \gamma (r(s)+V(f(s,a)))\\
%     & = \max_u (1-\gamma)r(s) + \gamma (r(s) + V(f(s,a)))
% \end{aligned}
% \end{equation}

%\begin{proof}[Proof of Proposition~\ref{prop:lower_bound}]
%From Theorem~\ref{thm:Bellman}, the contraction mapping implies that
%\begin{equation}
%    \| V^{(k+1)}(x) - V^*(x) \|_\infty \le \gamma \| V^{(k)}- V^*(x) \|_\infty.
%\end{equation}
%We then have
%\begin{equation}
%\begin{aligned}
%    \epsilon & \ge \| V^{(k+1)}(x)  -V^{(k)}(x) \|_\infty \\
%    & = \| V^{(k+1)}(x) - V^*(x) + V^*(x) - V^{(k)}(x) \|_\infty\\
%    & = \big\| \big( V^{(k+1)}(x)  - V^*(x)\big) - \big( V^{(k)}(x) - V^*(x) \big) \big\|_\infty  \\
%    &\ge \| V^{(k)}(x) - V^*(x)   \|_\infty - \| V^{(k+1)}(x) - V^*(x) \|_\infty\\
%    & \ge (\frac{1}{\gamma} -1)\| V^{(k+1)}(x) - V^*(x) \|_\infty\\
%    & = \frac{1-\gamma}{\gamma} \| V^{(k+1)}(x) - V^*(x) \|_\infty
%\end{aligned}
%\end{equation}
%which yields $V^{(k+1)}(x) - \frac{\gamma}{1-\gamma} \epsilon \le V^*(x)$, $\forall x$.
%\end{proof}

\balance
% \bibliographystyle{ieeetr}
% \printbibliography 
\bibliographystyle{IEEEtran}

\bibliography{references}

\end{document}